\documentclass{article}
\usepackage{amssymb}
\usepackage{amsmath,amsthm,amsfonts,longtable,amscd,mathrsfs}
\usepackage{color}
\usepackage[]{hyperref}
\usepackage{cite}
\usepackage{geometry}
\usepackage{changes}
\usepackage{array}
\usepackage{indentfirst}
\usepackage{booktabs}
\usepackage{tikz}
\usetikzlibrary{positioning, shapes.geometric}

\newcommand{\tabincell}[2]{\begin{tabular}{@{}#1@{}}#2\end{tabular}}
\evensidemargin=\oddsidemargin \textwidth=160mm \textheight=220mm
\renewcommand\arraystretch{0.69}

\newtheorem{definition}{Definition}
\newtheorem{lemma}{Lemma}
\newtheorem{theorem}{Theorem}
\newtheorem{corollary}{Corollary}
\newtheorem{example}{Example}
\newtheorem{remark}{Remark}

\begin{document}
\begin{sloppypar}
\begin{center}
{\LARGE \bf m-QMDS codes over mixed alphabets via orthogonal arrays}
\end{center}
\begin{center}
{\large Shanqi Pang, Mengqian Chen, Rong Yan, and Yan Zhu}\\
College of Mathematics and Information Science, Henan Normal University, Xinxiang 453007, China\\
Correspondence: Shanqi Pang, shanqipang@126.net
\end{center}
\noindent {\bf Abstract~---~}
The construction of quantum error-correcting codes (QECCs) with
good parameters is a hot topic in the area of quantum information and quantum computing. Quantum maximum distance separable (QMDS) codes are optimal because the minimum distance cannot be improved for a given length and code size. The QMDS codes over mixed alphabets are rarely known even if the existence and construction of QECCs over mixed alphabets with minimum distance more than or equal to three are still an open question. In this paper, we define an $m$-QMDS code over mixed alphabets, which is a generalization of QMDS codes. We establish a relation between $m$-QMDS codes over mixed alphabets and asymmetrical orthogonal arrays (OAs) with orthogonal partitions. Using this relation, we propose a general method to construct $m$-QMDS codes. As applications of this method, numerous infinite families of $m$-QMDS codes over mixed alphabets can be constructed explicitly. Compared with existing codes, the constructed codes have more flexibility in the choice of parameters, such as the alphabet sizes, length and dimension of the encoding state.

\noindent {\bf Keywords~---~}
Quantum error-correcting codes over mixed alphabets, uniform states, orthogonal arrays, orthogonal partitions, $m$-QMDS codes over mixed alphabets

\section{Introduction}
The no-go theorems like the no-cloning theorem \cite{Wootters82,Gisin97, Lamas02}, the no-broadcasting theorem \cite{Barnum96}, the no-deleting theorem \cite{Pati00}, and the no-hiding theorem \cite{Braunstein07} are consequences of the linearity and the unitarity of quantum mechanics. One of the problems with the feasibility of quantum computation is the difficulty of eliminating errors caused by inaccuracy and decoherence. Since the classical error-correcting techniques based on redundancy or repetition codes seemed to contradict the quantum no-cloning theorem, classical error-correcting codes can not be used in quantum computation. Quantum error-correcting codes will be necessary for preserving coherent states against noise and other unwanted interactions in quantum computation and communication \cite{Knill97}. Quantum systems are more fragile than classical systems. Errors are unavoidable when quantum information travels across a noisy channel \cite{shor95,Knill97,Knill00}. The primary tool to deal with different types of quantum noises is QECCs \cite{shor95, Knill97, Be96, Steane96}. They are the backbone of quantum fault-tolerant protocols needed to reliably operate scalable quantum computers and play an important role in quantum information tasks, such as quantum key distribution and entanglement purification \cite{gu24,Knill98,Gyorgy24,Bennett92,Liu21,Ambainis06,Duer07}. Since their discovery, the construction of QECCs has come a long way. The construction of QMDS codes over a single alphabet has been exhaustively investigated in the literature \cite{Hu08,fuyuan,hanxiao,yr23,Bierbrauer00,Chau97,Feng02,Grassl04,Ketkar06,Li08,Huber20}.
But we are often faced with a more complex situation where quantum information is encoded in quantum resources with different dimensions. In particular, we often use hybrid systems of different dimensions \cite{Yeyati07,Lansbergen08,Rabl10} to store, transmit, and process quantum information. Therefore, it is necessary to extend QECCs over a single alphabet to mixed alphabets \cite{Wang13}. However, QECCs over mixed alphabets are rarely studied. Other than some QECCs over mixed alphabets given in \cite{sf21,Wang13,yr23,chao24}, little is known about the existence of QECCs $((n,K,d+1))_{s_1s_2\cdots s_n}$, especially QMDS codes. There are two main reasons. First, the QMDS codes over mixed alphabets such as the codes with the dimension equal to one have not been
found so far. Second, it is difficult to construct such codes. Therefore, Pang et al. \cite{chao24} define a near quantum MDS (NQMDS) code and construct many NQMDS codes $((2d+1,1,d+1))_{s^{2d}2^1}$.

In this paper, we define an $m$-QMDS code over mixed alphabets, which is a generalization of QMDS codes and NQMDS codes. We establish a relation between $m$-QMDS codes over mixed alphabets and asymmetrical OAs with orthogonal partitions. Using this relation, we propose a general method to construct $m$-QMDS codes. As applications of this method, numerous infinite families of $m$-QMDS codes over mixed alphabets can be constructed explicitly. These quantum codes have a lot of advantages. They can detect all errors acting on at most $d$ subsystems and correct all errors acting on at most $\frac{d}{2}$ subsystems. It is possibly more convenient to design arithmetic and logic unit based on quantum technologies since each basis state of the quantum codes constructed is a uniform state which has far fewer terms and the coefficients of each term are all equal to 1.
As stated in \cite{Jin19,Simos18,Zhou09} the possible experimental realization techniques of near MDS codes will focus on secret sharing scheme. Since near MDS codes have similar properties to MDS codes, it is also possible that $m$-QMDS codes have similar properties to QMDS codes and will be applied to the construction of quantum secret sharing scheme.

The remainder of this paper is organized as follows. Section 2 provides symbols, concepts and lemmas, while Section 3 presents the main results. Some
conclusions are drawn in Section 4.

\section{Preliminaries}
To present our results, we first make some preparations.

$\mathbb{C}^s$ is a Hilbert space. Let $\mathbb{Z}_s^n$ denote the $n$ dimensional space over a ring $\mathbb{Z}_s=\{0,1,\ldots,s-1\}$. $\mathbb{F}_s$ denotes a Galois field.  $\mathbb{Z}_s=\mathbb{F}_s$ when $s$ is a prime power. Let $A^T$ be the transpose of the matrix $A$ and $\textbf{(s)}=(0,1,\ldots,s-1)^T$. Symbols $\textbf{0}_{r}$ and $\textbf{1}_{r}$ represent the $r\times 1$ vector of $0$s and $1$s, respectively. $\prod$ denotes continued-product. $C\subset \{1,2,\ldots,n\}$ means that $C$ is a subset of $\{1,2,\ldots,n\}$ and $C\neq \{1,2,\ldots,n\}$. $|C|$ denotes the number of elements in the set $C$.
$min\{a_1,a_2\}=\left\{
\begin{array}{l}
a_1, \ \text{if} \ a_1\leq a_2,\\
a_2, \ \text{if} \ a_1>a_2.
\end{array}
\right.$ $max\{a_1,a_2\}=\left\{
\begin{array}{l}
a_2, \ \text{if} \ a_1\leq a_2,\\
a_1, \ \text{if} \ a_1>a_2.
\end{array}
\right.$ If $A=(a_{ij})_{m\times n}$ and $B=(b_{ij})_{u\times v}$ have entries from a finite field with binary operations (usually addition $+$ multiplication $\cdot$), the Kronecker sum $A\oplus B$ is defined as $A\oplus B=(a_{ij}+ B)_{mu\times nv}$, where $a_{ij}+ B$ represents the $u\times v$ matrix with entries $a_{ij}+ b_{rs}$ $(1\leq r\leq u, 1\leq s\leq v)$, and the Kronecker product $A\otimes B$ is defined as $A\otimes B=(a_{ij}\cdot B)_{mu\times nv}$.

\begin{definition}\cite{hss}
An $r \times n$ matrix A, having $k_i$ columns with $s_i$ (levels), $i=1,\ldots, v$, $n=\sum \limits_{i=1}^v k_i$, $s_i\ne s_j$, for $i\neq j$, is called an orthogonal array $OA(r, n, s_1^{k_1}s_2^{k_2}\cdots s_v^{k_v}, t)$ of $t$ and size $r$, if each $r\times t$ submatrix of A contains all of the possible row vectors with the same frequency. An OA with $s_1=s_2=\cdots=s_v=s$ is called symmetric; otherwise, it is called asymmetric. If $\sum _{i=1}^vk_i(s_i-1)=r-1$, the ${\rm OA}(r,n,s_1^{k_1}s_2^{k_2}\cdots s_v^{k_v},2)$ is saturated.
\end{definition}

\begin{definition}\cite{hss}
Let $S^l=\{(v_1,\ldots,v_l)|v_i\in \mathbb{Z}_{s_i},i=1, 2,\ldots,l\}$. The Hamming distance HD(\textbf{u},\textbf{v}) between two vectors \textbf{u}$= (u_1,\ldots,u_l)$, \textbf{v}$= (v_1,\ldots,v_l) \in  S^l$ is defined as the number of positions in which they differ. The minimal distance of a matrix $A$, denoted by   $MD(A)$, is defined to be the minimal Hamming distance between its distinct rows. Let HD(A) denote all possible values of the Hamming distance between two distinct
rows of A.
\end{definition}

Let $\mathcal{A}$ be an additive group of $s$ elements. $\mathcal{A}^t$, $t\geq1$, denotes the additive group of order $s^t$ consisting of all $t$-tuples of entries from $\mathcal{A}$ with the usual vector addition as the binary operation. Let $\mathcal{A}_0^t=\{(x_1,x_2,\ldots,x_t):x_1=\cdots=x_t\in\mathcal{A}\}$. Then, $\mathcal{A}_0^t$ is a subgroup of $\mathcal{A}^t$ of order $s$, and its cosets will be denoted by $\mathcal{A}_i^t$, $i=1,\ldots,s^{t-1}-1$.

\begin{definition}\cite{hss}
An $r\times c$ matrix $D$ based on $\mathcal{A}$ is called a difference scheme of strength $t$ if, for every $r\times t$ submatrix, each set $\mathcal{A}_i^t$, $i=0,1,\ldots,s^{t-1}-1$, is represented equally often when the rows of the submatrix are viewed as elements of $\mathcal{A}^t$. Such a matrix is denoted by $D_t(r,c,s)$. Especially, $D_2(r,c,s)$ is short for $D(r, c, s)$.
\end{definition}

\begin{definition}\cite{wangjing}
Let $A$ be an $OA(r, n, s_1^{k_1}s_2^{k_2}\cdots s_v^{k_v}, t)$. Suppose the rows of $A$ can be partitioned into $u$ submatrices $A_{1}, A_{2}, \ldots, A_{u}$ such  that each $A_i$ is an $OA(\frac r u,n,s_1^{k_1}s_2^{k_2}\cdots s_v^{k_v},t_1)$ with $t_1\geq0$. Then the set $\{A_{1}, A_{2}, \ldots, A_{u}\}$ is called an orthogonal partition of strength $t_1$ of $A$.
\end{definition}

\begin{definition}\cite{sf21}\label{shifei}
A subspace of $\mathbb{C}^{s_1} \otimes \mathbb{C}^{s_2}\otimes \cdots \otimes\mathbb{C}^{s_n}$ $Q$ is an $((n,K,d+1))_{s_1,s_2,\ldots,s_n}$ $\rm QECC$ if for any $d$ parties, the reductions of all states in $Q$ to the $d$ parties are identical.
\end{definition}

Here $n$ is the number of qudits or the length of the code, $K$ is the dimension of the encoding state, $d+1$ is the minimum distance, and $s_1,s_2,\ldots,s_n$ are the alphabet sizes or $s_i$ indicates the dimensions of the $i$th particle.

\begin{definition}\cite{Wang13,sf21}\label{hunhejie}
An $((n,K,d+1))_{s_1 s_2\cdots s_n}$ QECC satisfies the quantum Singleton bound
$K\leq min \{
\prod_{j\in C}s_j|C\subset \{1,2,\ldots,n\}, |C|=n-2d
\}, \  for \ n\geq 2d+1, \ and \ K\leq 1, \ for \ n=2d.$
Specially, an $((n,K,d+1))_s$ QECC has the quantum Singleton bound $K\leq s^{n-2d}$.
\end{definition}

\begin{definition}\cite{chao24}
An $((n,K,d+1))_{s_1 s_2\cdots s_n}$ QECC is called an NQMDS code if
$K=min \{
\prod_{j\in C}s_j|C\subset \{1,2,\ldots,n\}, |C|=n-2d
\}-1, \ for \ n\geq 2d+1.$
\end{definition}

\begin{definition}\label{mQMDS}
An $((n,K,d+1))_{s_1 s_2\cdots s_n}$ QECC is called $m$-QMDS if
$K=min \{
\prod_{j\in C}s_j|C\subset \{1,2,\ldots,n\}, |C|=n-2d
\}-m, \ for \ n\geq 2d+1$.
\end{definition}

Obviously, an $m$-QMDS code is a QMDS code for $m=0$ and an $m$-QMDS code is an NQMDS code for $m=1$.

\begin{remark}\label{m}
If there exists an $((n,K,d+1))_{s_1 s_2\cdots s_n}$ QECC for
$K=min \{
\prod_{j\in C}s_j|C\subset \{1,2,\ldots,n\}, |C|=n-2d-1
\}, \ for \ n\geq 2d+1.$
We only consider the case
$0\leq m \leq min \{
\prod_{j\in C}s_j|C\subset \{1,2,\ldots,n\}, |C|=n-2d
\}
-min \{
\prod_{j\in C}s_j|C\subset \{1,2,\ldots,n\}, |C|=n-2d-1
\}$.
\end{remark}

\begin{lemma}\cite{hss,chenguangzhou}\label{chazhen}
(i) A difference scheme $D(s,s,s)$ and a difference scheme $D(2s,2s,s)$ exist for any prime power $s$. (ii) There is a difference matrix $D_3(s^2,4,s)$ over $\mathbb{Z}_s$ for any $s\geq2$.
\end{lemma}

\begin{lemma}\cite{chenguangzhou}\label{chazhen3}
Let $D$ be a difference scheme $D_t(r,c,s)$. Then $D\oplus \textbf{(s)}=OA(rs,c,s,t)$.
\end{lemma}

\begin{lemma}\cite{npj}\label{zuixiaojuli}
If $A$ is an $OA(s^t,n,s,t)$, then $MD(A)=n-t+1$.
\end{lemma}

\begin{lemma}\cite{zx21Q}\label{kuozhang}
(Expansive replacement method) Suppose $A$ is an OA of strength $t$ with column 1
having $s_1$ levels and that $B$ also is an OA of strength $t$ with $s_1$ rows. After making a one-to-one mapping
between the levels of column 1 in $A$ and the rows of $B$, if each level of column 1 in $A$ is replaced by the corresponding row from $B$, we can obtain an OA of strength $t$.
\end{lemma}

\begin{lemma}\cite{chao24}\label{ym}
Assume that an $OA(r,n,s_1s_2 \cdots s_n,t)$ exists with $MD=h$ and an orthogonal partition $\{A_1,\ldots,A_K\}$ of strength $t'$. Let $d+1=min\{t'+1,h\}$. Then, there exists an $((n,K,d+1))_{s_1s_2\cdots s_n}$ QECC.
\end{lemma}

\begin{lemma}\cite{cstm} \label{zx2}
Consider any two distinct rows of a saturated $OA(r,m_1+m_2,s_1^{m_1} s_2^{m_2},2)$, A. For $i=1,2$, assume $d_i$ is the Hamming
distance between these two rows of the $s_i$-level columns. Then $HD(A)=\{d_1+d_2: s_1d_1+s_2d_2=r,0\leq d_i\leq m_i, i=1,2\}$. Particularly, if $r=s_1^2$, then
$HD(A)=\{\frac{s_1^2+(s_2-s_1)(m_1-1)}{s_2},
\frac{s_1^2+(s_2-s_1)m_1}{s_2}\}$.
\end{lemma}

\begin{lemma}\cite{hss}\label{bush}
(i) If $s\geq 2$ is a prime power then an $OA(s^t,s+1,s,t)$ of index unity exists whenever $s\geq t-1\geq 0$. (ii) If $s=2^m$, $m\geq1$, and $t=3$ then there exists an $OA(s^3,s+2,s,t)$.
\end{lemma}

\begin{lemma}\label{chengtui}
Assume that A is an $OA(r_1, n, s_1s_2\cdots s_n, t)$ with $MD(A)=h_1$, and that B is an $OA(r_2, n, q_1q_2\cdots q_n, t)$ with $MD(B)=h_2$. Let $h=min\{h_1, h_2\}$. Then, there exists an $OA(r_1r_2,n,(s_1q_1)(s_2q_2)\cdots (s_nq_n), t)$ with $MD=h$.
\end{lemma}
\begin{proof}
The proof is similar to the Lemma 7 in \cite{yr23}.
\end{proof}

\begin{lemma}\cite{hss}\label{qu1}
Taking the runs in an $OA(r,n,s,t)$ that begin with 0 (or any other particular symbol) and omitting the first column yields an
$OA(r/s,n-1,s,t-1)$. If we assume that these are the initial runs, the process can be represented by the following diagram:
\begin{center}
\begin{tabular}{|c|c|}\hline
0&\\
\vdots&OA(r/s,n-1,s,t-1)\\
0&\\\hline
1&\\
\vdots&OA(r/s,n-1,s,t-1)\\
1&\\\hline
&\\
\vdots&\vdots\\
&\\\hline
\end{tabular}
\end{center}
\end{lemma}

\section{New $m$-QMDS code from OAs}
In this section, we construct several infinite classes of $m$-QMDS codes over mixed alphabets from asymmetrical OAs. Compared with the NQMDS codes $((2d+1,1,d+1))_{s^{2d}2^1}$ in \cite{chao24}, the constructed codes have higher flexibility in length, dimension of the encoding state and alphabet sizes.

By using difference schemes and the expansive replacement method, we have the following result.

\begin{theorem}\label{5s2}
An $(s-1)$-QMDS code $((4+n,1,3))_{(s^2)^1 s^3 s_1^1 s_2^1\cdots s_n^1}$ exists for $s\geq 2$ and $s=s_1s_2\cdots s_n$.
\end{theorem}

\begin{proof}
By Lemma \ref{chazhen} (ii), there exists a difference matrix $D=D_3(s^2,4,s)$ for $s\geq2$.
Then $A=D\oplus\textbf{(s)}$ is an $OA(s^3,4,s,3)$ with $MD(A)=2$ from Lemmas \ref{chazhen3} and \ref{zuixiaojuli}. We construct an array $B=OA(s^3,5,(s^2)^1 s^4,2)$ as follows.
\begin{align*}
B=(\textbf{(s}^\textbf{2}\textbf{)}\otimes \textbf{1}_s,D\oplus\textbf{(s)})
=\left(\begin{array}{cc}
0\otimes \textbf{1}_s&d_1\oplus\textbf{(s)}\\
1\otimes \textbf{1}_s&d_2\oplus\textbf{(s)}\\
\vdots&\vdots\\
(s^2-1)\otimes \textbf{1}_s&d_{s^2}\oplus\textbf{(s)}\\
\end{array}\right)=\left(\begin{array}{c}
B_1\\
B_2\\
\vdots\\
B_{s^2}\\
\end{array}\right)
\end{align*}
where $d_i$ is the $i$ row of $D$ for $1\leq i\leq s^2$.
Apparently, $MD(B_i)=4$, $MD(B_{i_1},B_{i_2})=1+2=3$ for $1\leq i_1 < i_2 \leq s^2$ from $MD(A)=2$. Then $MD(B)=3$. An $OA(s^3,4+n,(s^2)^1 s^3 s_1^1s_2^1\cdots s_n^1,2)$ with $MD=3$ for $s=s_1s_2\cdots s_n$ and $s\geq4$ can be obtained after an $s$-level column of $B$ is replaced by an $OA(s,n,s_1^1 s_2^1 \cdots s_n^1,n)$ from Lemma \ref{kuozhang}. Therefore, an $(s-1)$-QMDS code $((4+n,1,3))_{(s^2)^1 s^3 s_1^1 s_2^1\cdots s_n^1}$ exists for $s\geq 2$ and $s=s_1s_2\cdots s_n$ from Lemma \ref{ym} and Definition \ref{mQMDS}.
\end{proof}

Furthermore, by using difference schemes, multiplication of OAs and the expansive replacement method, we have the following result.
\begin{theorem}\label{52s}
An $(s-1)$-QMDS code $((4+n,1,3))_{(2s)^1 s^3 s_1^1 s_2^1 \cdots s_n^1}$ exists for $s\geq 2$ and $s=s_1s_2\cdots s_n$.
\end{theorem}

\begin{proof}
By Lemma \ref{chazhen} (i), there exists a difference matrix $D(2s,2s,s)$ for any prime power $s$. Then there exists an $A=(\textbf{(2s)}\otimes \textbf{1}_s,D(2s,2s,s)\oplus\textbf{(s)})$ is a saturated $OA(2s^2,2s+1,(2s)^1 s^{2s},2)$. We have $MD(A)=2s-1$ from Lemma \ref{zx2}. For any prime power $s$, we can obtain an $OA(2s^2,5,(2s)^1 s^4,2)$ with $MD=3$ by deleting $2s-4$ $s$-level columns from $A$. Now let's prove there exists an $OA(2s^2,5,(2s)^1 s^4,2)$ with $MD=3$ for
$s=u_1^{l_1}u_2^{l_2}\cdots u_v^{l_v}$ where $u_1,u_2,\ldots,u_v$ are all primes, $v\geq2$, $2\leq u_1<u_2<\cdots<u_v\leq s$ and $l_1,l_2,\ldots,l_v\geq 1$. We consider the following three cases:

(1) When $u_1> 2$, there exists an $OA(2(u_1^{l_1})^2,5,(2u_1^{l_1})^1 (u_1^{l_1})^4,2)$ with $MD=3$.
There exists an $OA((u_j^{l_j})^2,5,u_j^{l_j},2)$ with $MD\geq3$ for $2\leq j\leq v$ from Lemmas \ref{bush} and \ref{zuixiaojuli}. By Lemma \ref{chengtui}, there exists an $OA(2s^2,5,(2s)^1 s^4,2)$ with $MD=3$.

(2) When $u_1=2$ and $u_2=3$, we need to prove that there exists an $OA(2s^2,5,(2s)^1 s^4,2)$ with $MD=3$. We consider the following three cases:

(a) When $l_1=l_2=1$, there exists an $OA(72,5,12^1 6^4,2)$ with $MD=3$ in \cite{cz} and an $OA((u_j^{l_j})^2,5,u_j^{l_j},2)$ with $MD\geq3$ for $3\leq j\leq v$ from Lemmas \ref{bush} and \ref{zuixiaojuli}.
(b) When $l_1=1$ and $l_2\geq2$, there exists an $OA(8,5,4^1 2^4,2)$ with $MD=3$ and an $OA((u_j^{l_j})^2,5,u_j^{l_j},2)$ with $MD\geq3$ for $2\leq j\leq v$ from Lemmas \ref{bush} and \ref{zuixiaojuli}.
(c) When $l_1\geq2$ and $l_2=1$, there exists an $OA(2u_2^2,5,(2u_2)^1 u_2^4,2)=OA(18,5,6^1 3^4,2)$ with $MD=3$ and an $OA((u_j^{l_j})^2,5,u_j^{l_j},2)$ with $MD\geq3$ for $j=1$ and $3\leq j\leq v$ from Lemmas \ref{bush} and \ref{zuixiaojuli}. Then there exists an $OA(2s^2,5,(2s)^1 s^4,2)$ with $MD=3$ from Lemma \ref{chengtui}.

(3) When $u_1=2$ and $u_2>3$, there exists an $OA(2(u_1^{l_1})^2,5,(2u_1^{l_1})^1 (u_1^{l_1})^4,2)$
with $MD=3$ for $l_1\geq1$  and an $OA((u_j^{l_j})^2,5,u_j^{l_j},2)$ with $MD\geq3$ for $2\leq j\leq v$ from Lemmas \ref{bush} and \ref{zuixiaojuli}. By Lemma \ref{chengtui}, there exists an $OA(2s^2,5,(2s)^1 s^4,2)$ with $MD=3$.

So an $OA(2s^2,4+n,(2s)^1 s^3 s_1^1s_2^1\cdots s_n^1,2)$ with $MD=3$ for $s=s_1s_2\cdots s_n$ and $s\geq2$ can be obtained after an $s$-level column of $OA(2s^2,5,(2s)^1 s^4,2)$ is replaced by an $OA(s,n,s_1^1 s_2^1 \cdots s_n^1,n)$ from Lemma \ref{kuozhang}. Therefore, an $(s-1)$-QMDS code $((4+n,1,3))_{(2s)^1 s^3 s_1^1 s_2^1\cdots s_n^1}$ exists for $s\geq 2$ and $s=s_1s_2\cdots s_n$ from Lemma \ref{ym} and Definition \ref{mQMDS}.
\end{proof}

Table I lists plenty of $m$-QMDS codes including many infinite classes of such codes constructed by Theorems \ref{5s2} and \ref{52s}.

\noindent {TABLE I.} New $m$-QMDS codes $((4+n,1,3))_{(s^2)^1 s^3 s_1^1  s_2^1\cdots s_n^1}$ and $((4+n,1,3))_{(2s)^1 s^3 s_1^1  s_2^1\cdots s_n^1}$ with $m=s-1$ from Theorems \ref{5s2} and \ref{52s}, respectively.
\begin{center}
\vspace{-7mm}
\renewcommand\arraystretch{1}
\setlength{\tabcolsep}{7pt}
\begin{longtable}{lllll}
\label{Table} \\
\hline\hline $s$&$s_1s_2\cdots s_n$&$((4+n,1,3))_{(s^2)^1 s^3 s_1^1 \cdots s_n^1}$ by Th1&$((4+n,1,3))_{(2s)^1 s^3 s_1^1 \cdots s_n^1}$ by Th2&$m$\\  \hline
\endfirsthead
\multicolumn{5}{c}
{{\bfseries }} \\
\hline\hline $s$&$s_1s_2\cdots s_n$&$((4+n,1,3))_{(s^2)^1 s^3 s_1^1 \cdots s_n^1}$ by Th1&$((4+n,1,3))_{(2s)^1 s^3 s_1^1 \cdots s_n^1}$ by Th2&$m$\\  \hline
\endhead
\hline \\
\endfoot
\endlastfoot
2&$2$&$((5,1,3))_{4^1 2^4}$&$((5,1,3))_{4^1 2^4}$&1\\
3&$3$&$((5,1,3))_{9^1 3^4}$&$((5,1,3))_{6^1 3^4}$&2\\
4&$4$&$((5,1,3))_{16^1 4^4}$&$((5,1,3))_{8^1 4^4}$&3\\
4&$2\times2$&$((6,1,3))_{16^1 4^3 2^2}$&$((6,1,3))_{8^1 4^3 2^2}$&3\\
5&$5$&$((5,1,3))_{25^1 5^4}$&$((5,1,3))_{10^1 5^4}$&4\\
6&$6$&$((5,1,3))_{36^1 6^4}$&$((5,1,3))_{12^1 6^4}$&5\\
6&$2\times3$&$((6,1,3))_{36^1 6^3 3^1 2^1}$&$((6,1,3))_{12^1 6^3 3^1 2^1}$&5\\
7&$7$&$((5,1,3))_{49^1 7^4}$&$((5,1,3))_{14^1 7^4}$&6\\
8&$8$&$((5,1,3))_{64^1 8^4}$&$((5,1,3))_{16^1 8^4}$&7\\
8&$2\times4$&$((6,1,3))_{64^1 8^3 4^1 2^1}$&$((6,1,3))_{16^1 8^3 4^1 2^1}$&7\\
8&$2\times2\times2$&$((7,1,3))_{64^1 8^3 2^3}$&$((7,1,3))_{16^1 8^3 2^3}$&7\\
9&$9$&$((5,1,3))_{81^1 9^4}$&$((5,1,3))_{18^1 9^4}$&8\\
9&$3\times3$&$((6,1,3))_{81^1 9^3 3^2}$&$((6,1,3))_{18^1 9^3 3^2}$&8\\

10&$10$&$((5,1,3))_{100^1 10^4}$&$((5,1,3))_{20^1 10^4}$&9\\
10&$2\times5$&$((6,1,3))_{100^1 10^3 5^1 2^1}$&$((6,1,3))_{20^1 10^3 5^1 2^1}$&9\\
11&$11$&$((5,1,3))_{121^1 11^4}$&$((5,1,3))_{22^1 11^4}$&10\\
12&$12$&$((5,1,3))_{144^1 12^4}$&$((5,1,3))_{24^1 12^4}$&11\\
12&$2\times6$&$((6,1,3))_{144^1 12^3 6^1 2^1}$&$((6,1,3))_{24^1 12^3 6^1 2^1}$&11\\
12&$3\times4$&$((6,1,3))_{144^1 12^3 4^1 3^1}$&$((6,1,3))_{24^1 12^3 4^1 3^1}$&11\\
$s\geq 13$&$s_1 s_2 \cdots s_n$&$((4+n,1,3))_{(s^2)^1 s^3 s_1^1 s_2^1 \cdots s_n^1}$&$((4+n,1,3))_{(2s)^1 s^3 s_1^1 s_2^1 \cdots s_n^1}$&$s-1$\\\hline\hline
\end{longtable}
\end{center}

\vspace{-10mm}
Theorems \ref{s1QMDS} and \ref{tn} are based on the existence of an OA.
\begin{theorem}\label{s1QMDS}
If there exists an $OA(s^d,2d,s,d)$, then there exists an $(s_1-1)$-QMDS code $((2d+1,1,d+1))_{s^{2d-1} (\frac{s}{s_1})^1 s_1^1}$ for $s_1|s$ and $s\geq s_1^2$.
\end{theorem}
\begin{proof}
By Lemma \ref{zuixiaojuli}, the minimal distance of $OA(s^d,2d,s,d)$ is $d+1$, We can obtain an $A=OA(s^d,2d+1,s^{2d-1} (\frac{s}{s_1})^1 s_1^1,d)$ with $MD(A)=d+1$ after an $s$-level column of $OA(s^d,2d,s,d)$ is replaced by an $OA(s,2,(\frac{s}{s_1})^1 s_1^1,2)$ from Lemma \ref{kuozhang}. By Lemma \ref{ym} and Definition \ref{mQMDS}, there exists an $(s_1-1)$-QMDS code $((2d+1,1,d+1))_{s^{2d-1} (\frac{s}{s_1})^1 s_1^1}$.
\end{proof}
The following corollary follows immediately from Theorem \ref{s1QMDS}.
\begin{corollary}\label{1wei}
If $s_1|s$ and $s\geq s_1^2$, then

(i) an $(s_1-1)$-QMDS code $((3,1,2))_{s^1 (\frac{s}{s_1})^1 s_1^1}$ exists for $s\geq4$.

(ii) an $(s_1-1)$-QMDS code $((5,1,3))_{s^3 (\frac{s}{s_1})^1 s_1^1}$ exists for $s\geq4$ and $s\neq 6$.

(iii) an $(s_1-1)$-QMDS code $((7,1,4))_{s^5 (\frac{s}{s_1})^1 s_1^1}$ exists with  $s=u_1u_2\cdots u_v$ for prime powers $u_1,u_2,\ldots,u_v$ and $min\{u_1,u_2,\ldots,u_v\}\geq 4$.

(iv) an $(s_1-1)$-QMDS code $((2d+1,1,d+1))_{s^{2d-1} (\frac{s}{s_1})^1 (s_1)^1}$ exists with $d\geq 4$ and that $s=u_1u_2\cdots u_v$ for prime powers $u_1,u_2,\ldots,u_v$ and $min\{u_1,u_2,\ldots,u_v\}\geq 2d-1$. .
\end{corollary}

\begin{proof}
(i) By Lemma \ref{bush}, there exists an $OA(s,2,s,1)$ for $s\geq4$.

(ii) From Lemma \ref{bush}, an $OA(s^2,4,s,2)$ exists for a prime power $s$. When $s=10$, we have an $OA(100,4,10,2)$ in \cite{cz}. From \cite{hss}, we have the following conclusions: (1) An $OA(s^2,n_1,s,2)$ exists if and only if $n_1-2$ pairwise orthogonal Latin squares of order $s$ exist. (2) There exist more than 2 pairwise orthogonal Latin squares of order $s\geq 12$ which is not a prime power. Then there exists an $OA(s^2,4,s,2)$ for $s\geq4$ and $s\neq 6$.

(iii) By Lemma \ref{bush}, there exists an $OA(u_i^3,6,u_i,3)$ with a prime power $u_i\geq4$ for $1\leq i\leq v$. By Lemma \ref{chengtui}, there exists an $OA(s^3,6,s,3)$ with $s=u_1u_2\cdots u_v$ for prime powers $u_1,u_2,\ldots,u_v$ and $min\{u_1,u_2,\ldots,u_v\}\geq 4$.

(iv) By Lemma \ref{bush}, there exists an $OA(u_i^d,2d,u_i,d)$ with a prime power $u_i\geq 2d-1$ for $1\leq i\leq v$. By Lemma \ref{chengtui}, there exists an $OA(s^d,2d,s,d)$ with $s=u_1u_2\cdots u_v$ for prime powers $u_1,u_2,\ldots,u_v$ and $min\{u_1,u_2,\ldots,u_v\}\geq 2d-1$.

It follows from Theorem \ref{s1QMDS} that the corollary holds.
\end{proof}

\noindent {TABLE II.} New 1-QMDS codes $((2d+1,1,d+1))_{s^{2d-1} (\frac{s}{2})^1 2^1}$ from Corollary \ref{1wei}.
\vspace{-7mm}
\renewcommand\arraystretch{1}
\setlength{\tabcolsep}{7pt}
\begin{longtable}{lll}
\label{Table} \\
\hline\hline $d$&$s$&$((2d+1,1,d+1))_{s^{2d-1} (\frac{s}{2})^1 2^1}$\\  \hline
\endfirsthead
\multicolumn{3}{c}
{{\bfseries }} \\
\hline\hline $d$&$s$&$((2d+1,1,d+1))_{s^{2d-1} (\frac{s}{2})^1 2^1}$\\  \hline
\endhead
\hline \\
\endfoot
\endlastfoot
1&4&$((3,1,2))_{4^1 2^2}$\\
1&6&$((3,1,2))_{6^1 3^1 2^1}$\\
1&8&$((3,1,2))_{8^1 4^1 2^1}$\\
1&10&$((3,1,2))_{10^1 5^1 2^1}$\\
1&12&$((3,1,2))_{12^1 6^1 2^1}$\\
1&14&$((3,1,2))_{14^1 7^1 2^1}$\\
1&$s=2k(k\geq8)$&$((3,1,2))_{s^1 (\frac{s}{2})^1 2^1}$\\\hline

2&4&$((5,1,3))_{4^3 2^2}$\\
2&8&$((5,1,3))_{8^3 4^1 2^1}$\\
2&10&$((5,1,3))_{10^3 5^1 2^1}$\\
2&12&$((5,1,3))_{12^3 6^1 2^1}$\\
2&16&$((5,1,3))_{16^3 8^1 2^1}$\\
2&18&$((5,1,3))_{18^3 9^1 2^1}$\\
2&$s=2k(k\geq10)$&$((5,1,3))_{s^3 (\frac{s}{2})^1 2^1}$\\\hline

3&8&$((7,1,4))_{8^5 4^1 2^1}$\\
3&16&$((7,1,4))_{16^5 8^1 2^1}$\\
3&20&$((7,1,4))_{20^5 10^1 2^1}$\\
3&28&$((7,1,4))_{28^5 14^1 2^1}$\\
3&32&$((7,1,4))_{32^5 16^1 2^1}$\\
3&36&$((7,1,4))_{36^5 18^1 2^1}$\\
3&40&$((7,1,4))_{40^5 20^1 2^1}$\\
3&44&$((7,1,4))_{44^5 22^1 2^1}$\\
3&52&$((7,1,4))_{52^5 26^1 2^1}$\\
3&56&$((7,1,4))_{56^5 28^1 2^1}$\\
3&\tabincell{l}{$s=4\times 2^p \times u_1\times u_2\times \cdots \times u_v$\\($p\geq0$, $u_i$ is a prime power and $u_i\geq 5$)}&$((7,1,4))_{s^5 (\frac{s}{2})^1 2^1}$\\\hline

4&8&$((9,1,5))_{8^7 4^1 2^1}$\\
4&16&$((9,1,5))_{16^7 8^1 2^1}$\\
4&32&$((9,1,5))_{32^7 16^1 2^1}$\\
4&56&$((9,1,5))_{56^7 28^1 2^1}$\\
4&64&$((9,1,5))_{64^7 32^1 2^1}$\\
4&72&$((9,1,5))_{72^7 36^1 2^1}$\\
4&88&$((9,1,5))_{88^7 44^1 2^1}$\\
\tabincell{l}{4 \\ }&\tabincell{l}{$s=8\times 2^p \times u_1\times u_2\times \cdots \times u_v$\\($p\geq0$, $u_i$ is a prime power and $u_i\geq 7$)}&\tabincell{l}{$((9,1,5))_{s^7 (\frac{s}{2})^1 2^1}$ \\ }\\\hline

5&16&$((11,1,6))_{16^9 8^1 2^1}$\\
5&32&$((11,1,6))_{32^9 16^1 2^1}$\\
5&64&$((11,1,6))_{64^9 32^1 2^1}$\\
5&128&$((11,1,6))_{128^9 64^1 2^1}$\\
5&144&$((11,1,6))_{144^9 72^1 2^1}$\\
5&\tabincell{l}{$s=16\times 2^p \times u_1\times u_2\times \cdots \times u_v$\\($p\geq0$, $u_i$ is a prime power and $u_i\geq 9$)}&$((11,1,6))_{s^7 (\frac{s}{2})^1 2^1}$\\\hline
$\cdots$&$\cdots$&$\cdots$\\\hline\hline
\end{longtable}

\vspace{-5mm}
\noindent {TABLE III.} New 2-QMDS codes $((2d+1,1,d+1))_{s^{2d-1} (\frac{s}{3})^1 3^1}$ from Corollary \ref{1wei}.
\vspace{-7mm}
\renewcommand\arraystretch{1}
\setlength{\tabcolsep}{7pt}
\begin{longtable}{lll}
\label{Table} \\
\hline\hline $d$&$s$&$((2d+1,1,d+1))_{s^{2d-1} (\frac{s}{3})^1 3^1}$\\  \hline
\endfirsthead
\multicolumn{3}{c}
{{\bfseries }} \\
\hline\hline $d$&$s$&$((2d+1,1,d+1))_{s^{2d-1} (\frac{s}{3})^1 3^1}$\\  \hline
\endhead
\hline \\
\endfoot
\endlastfoot
1&9&$((3,1,2))_{9^1 3^2}$\\
1&12&$((3,1,2))_{12^1 4^1 3^1}$\\
1&15&$((3,1,2))_{15^1 5^1 3^1}$\\
1&18&$((3,1,2))_{18^1 6^1 3^1}$\\
1&21&$((3,1,2))_{21^1 7^1 3^1}$\\
1&$s=3k(k\geq8)$&$((3,1,2))_{s^1 (\frac{s}{3})^1 3^1}$\\\hline

2&9&$((5,1,3))_{9^3 3^2}$\\
2&12&$((5,1,3))_{12^3 4^1 3^1}$\\
2&15&$((5,1,3))_{15^3 5^1 3^1}$\\
2&18&$((5,1,3))_{18^3 6^1 3^1}$\\
2&21&$((5,1,3))_{21^3 7^1 3^1}$\\
2&$s=3k(k\geq8)$&$((5,1,3))_{s^3 (\frac{s}{3})^1 3^1}$\\\hline

3&9&$((7,1,4))_{8^5 4^1 2^1}$\\
3&27&$((7,1,4))_{27^5 9^1 3^1}$\\
3&36&$((7,1,4))_{36^5 12^1 3^1}$\\
3&45&$((7,1,4))_{45^5 15^1 3^1}$\\
3&63&$((7,1,4))_{63^5 21^1 3^1}$\\
3&\tabincell{l}{$s=9\times 3^p \times u_1\times u_2\times \cdots \times u_v$\\($p\geq0$, $u_i$ is a prime power and $u_i\geq 4$)}&$((7,1,4))_{s^5 (\frac{s}{3})^1 3^1}$\\\hline

4&9&$((9,1,5))_{9^7 3^2}$\\
4&27&$((9,1,5))_{27^7 9^1 3^1}$\\
4&63&$((9,1,5))_{63^7 21^1 3^1}$\\
4&72&$((9,1,5))_{72^7 24^1 3^1}$\\
4&81&$((9,1,5))_{81^7 27^1 3^1}$\\
4&99&$((9,1,5))_{99^7 33^1 3^1}$\\
4&117&$((9,1,5))_{117^7 39^1 3^1}$\\
\tabincell{l}{4}&\tabincell{l}{$s=9\times 3^p \times u_1\times u_2\times \cdots \times u_v$\\($p\geq0$, $u_i$ is a prime power and $u_i\geq 7$)}&\tabincell{l}{$((9,1,5))_{s^7 (\frac{s}{2})^1 2^1}$}\\\hline

5&9&$((11,1,6))_{9^9 3^2}$\\
5&27&$((11,1,6))_{27^9 9^1 3^1}$\\
5&81&$((11,1,6))_{81^9 27^1 3^1}$\\
5&99&$((11,1,6))_{99^9 33^1 3^1}$\\
5&117&$((11,1,6))_{117^9 39^1 3^1}$\\
5&144&$((11,1,6))_{144^9 48^1 3^1}$\\
5&153&$((11,1,6))_{153^9 51^1 3^1}$\\
5&\tabincell{l}{$s=9\times 3^p \times u_1\times u_2\times \cdots \times u_v$\\($p\geq0$, $u_i$ is a prime power and $u_i\geq 9$)}&$((11,1,6))_{s^7 (\frac{s}{2})^1 2^1}$\\\hline
$\cdots$&$\cdots$&$\cdots$\\\hline\hline
\end{longtable}
\vspace{-5mm}

The following theorem implies that infinite classes of $m$-QMDS codes over mixed alphabets with more dimensions exist.
\begin{theorem}\label{tn}
Suppose an $OA(s^{d+l},2d+2l+1,s,d+l)$ exists for integer $d\geq1$ and that $s_1,s_2,\ldots,s_{n_1},q_1,q_2,\ldots,q_{n_2}\geq2$ are all integers. If $(s_1s_2\cdots s_{n_1})|s$ and $s=q_1 q_2 \cdots q_{n_2}$, then

(i) when $l\geq1$, there exist two $(s_1s_2\cdots s_{n_1}-1)s^l$-QMDS codes $((2d+l+n_1,s^l,d+1))_{s^{2d+l} s_1^1 s_2^1 \cdots s_{n_1}^1}$ and $((2d+l-1+n_1+n_2,s^l,d+1))_{s^{2d+l-1} s_1^1 s_2^1 \cdots s_{n_1}^1 q_1^1q_2^1\cdots q_{n_2}^1}$.

(ii) when $l=0$, there exists an $(s_1s_2\cdots s_{n_1}-1)$-QMDS code $((2d+n_1,1,d+1))_{s^{2d} s_1^1 s_2^1 \cdots s_{n_1}^1}$ and an $(\frac{s s_1s_2\cdots s_{n_1}}{w}-1)$-QMDS code $((2d-1+n_1+n_2,1,d+1))_{s^{2d-1} s_1^1 s_2^1 \cdots s_{n_1}^1 q_1^1q_2^1\cdots q_{n_2}^1}$ for $w=max\{s_1,s_2,\ldots,s_{n_1},q_1,q_2,\ldots,q_{n_2}\}$.
\end{theorem}
\begin{proof}
Let $A=OA(s^{d+l},2d+2l+1,s,d+l)$. By Lemma \ref{zuixiaojuli}, then $MD(A)=d+l+2$. By row permutations,
$A=\begin{pmatrix}
\mathbb{Z}_s^l\otimes \textbf{1}_{s^d}&\overline{A}
\end{pmatrix}=\begin{pmatrix}
(00\cdots0)\otimes \textbf{1}_{s^d}&A_1\\
(00\cdots1)\otimes \textbf{1}_{s^d}&A_2\\
\vdots&\vdots\\
(s-1,s-1,\cdots,s-1)\otimes \textbf{1}_{s^d}&A_{s^l}
\end{pmatrix}$. By Lemma \ref{qu1}, $A_1,A_2,\ldots,A_{s^l}$ are all $OA(s^d,2d+l+1,s,d)$, and $MD(\overline{A})=d+2$. We can obtain a $B=OA(s^{d+l},2d+l+1,s^{2d+l}s_1^1s_2^1\cdots s_{n_1}^1,d+l)$ with an orthogonal partition $\{B_1,B_2,\ldots,B_{s^l}\}$ of strength $d$ after an $s$-level column of $\overline{A}=OA(s^{d+l},2d+l+1,s,d+l)$ is replaced by an $OA(s,n_1,s_1^1s_2^1\cdots s_{n_1}^1,n_1)$ from Lemma \ref{kuozhang}. And $MD(B)=d+1$. When $l\geq0$, then there exists an $(s_1s_2\cdots s_{n_1}-1)s^l$-QMDS code $((2d+l+n_1,s^l,d+1))_{s^{2d+l} s_1^1 s_2^1\cdots s_{n_1}^1}$ from Lemma \ref{ym} and Definition \ref{mQMDS}. Similarly, We can obtain a $C=OA(s^{d+l},2d+l-1+n_1+n_2,s^{2d+l-1}s_1^1s_2^1\cdots s_{n_1}^1 q_1^1 q_2^1\cdots q_{n_2}^1,d+l)$ with $MD(C)=d+1$ after an $s$-level column of $B$ is replaced by an $OA(s,n_2,q_1^1 q_2^1\cdots q_{n_2}^1,n_2)$ from Lemma \ref{kuozhang}. And $C$ has an orthogonal partition $\{C_1,C_2,\ldots,C_{s^l}\}$ of strength $d$. By Lemma \ref{ym} and Definition \ref{mQMDS}, when $l\geq1$, then there exists an $(s_1s_2\cdots s_{n_1}-1)s^l$-QMDS code  $((2d+l-1+n_1+n_2,s^l,d+1))_{s^{2d+l-1} s_1^1 s_2^1 \cdots s_{n_1}^1 q_1^1 q_2^1 \cdots q_{n_2}^1}$, when $l=0$, then there exists an $(\frac{s s_1s_2\cdots s_{n_1}}{w}-1)$-QMDS code $((2d-1+n_1+n_2,1,d+1))_{s^{2d-1} s_1^1 s_2^1 \cdots s_{n_1}^1 q_1^1q_2^1\cdots q_{n_2}^1}$ for $w=max\{s_1,s_2,\ldots,s_{n_1},q_1,q_2,\ldots,q_{n_2}\}$.
\end{proof}

\begin{example}
Take $s=12$, $d=1$, $l=1$ in Theorem \ref{tn} (i). We can obtain all codes in Table IV and their orthogonal bases. For example, the code $((4,12,2))_{12^3 2^1}$ has an orthogonal base $\{|\phi_1\rangle,\ldots,|\phi_{12}\rangle\}$ where
$|\phi_1\rangle=|0,0,0,0\rangle+
|1,6,3,1\rangle+
|2,8,6,0\rangle+
|3,2,1,1\rangle+
|4,7,9,0\rangle+
|5,1,11,1\rangle+
|6,9,2,0\rangle+
|7,11,8,1\rangle+
|8,4,5,0\rangle+
|9,10,4,1\rangle+
|10,5,7,1\rangle+
|11,3,10,0\rangle$,\\
$|\phi_2\rangle=|0,2,6,1\rangle+
|1,1,1,0\rangle+
|2,7,4,1\rangle+
|3,9,7,0\rangle+
|4,3,2,1\rangle+
|5,8,10,0\rangle+
|6,4,11,0\rangle+
|7,10,3,0\rangle+
|8,6,9,1\rangle+
|9,5,0,0\rangle+
|10,11,5,1\rangle+
|11,0,8,1\rangle$,\\
$|\phi_3\rangle=|0,9,11,0\rangle+
|1,3,7,1\rangle+
|2,2,2,0\rangle+
|3,8,5,1\rangle+
|4,10,8,0\rangle+
|5,4,3,1\rangle+
|6,1,9,1\rangle+
|7,5,6,0\rangle+
|8,11,4,0\rangle+
|9,7,10,1\rangle+
|10,0,1,0\rangle+
|11,6,0,1\rangle$,\\
$|\phi_4\rangle=|0,5,4,1\rangle+
|1,10,6,0\rangle+
|2,4,8,1\rangle+
|3,3,3,0\rangle+
|4,9,0,1\rangle+
|5,11,9,0\rangle+
|6,7,1,1\rangle+
|7,2,10,1\rangle+
|8,0,7,0\rangle+
|9,6,5,0\rangle+
|10,8,11,1\rangle+
|11,1,2,0\rangle$,\\
$|\phi_5\rangle=|0,6,10,0\rangle+
|1,0,5,1\rangle+
|2,11,7,0\rangle+
|3,5,9,1\rangle+
|4,4,4,0\rangle+
|5,10,1,1\rangle+
|6,2,3,0\rangle+
|7,8,2,1\rangle+
|8,3,11,1\rangle+
|9,1,8,0\rangle+
|10,7,0,0\rangle+
|11,9,6,1\rangle$,\\
$|\phi_6\rangle=|0,11,2,1\rangle+
|1,7,11,0\rangle+
|2,1,0,1\rangle+
|3,6,8,0\rangle+
|4,0,10,1\rangle+
|5,5,5,0\rangle+
|6,10,7,1\rangle+
|7,3,4,0\rangle+
|8,9,3,1\rangle+
|9,4,6,1\rangle+
|10,2,9,0\rangle+
|11,8,1,0\rangle$,\\
$|\phi_7\rangle=|0,3,8,1\rangle+
|1,5,2,0\rangle+
|2,10,11,1\rangle+
|3,4,10,0\rangle+
|4,11,1,0\rangle+
|5,9,4,1\rangle+
|6,6,6,1\rangle+
|7,0,9,0\rangle+
|8,2,0,1\rangle+
|9,8,7,0\rangle+
|10,1,3,1\rangle+
|11,7,5,0\rangle$,\\
$|\phi_8\rangle=|0,10,5,1\rangle+
|1,4,9,1\rangle+
|2,0,3,0\rangle+
|3,11,6,1\rangle+
|4,5,11,0\rangle+
|5,6,2,0\rangle+
|6,8,0,0\rangle+
|7,7,7,1\rangle+
|8,1,10,0\rangle+
|9,3,1,1\rangle+
|10,9,8,0\rangle+
|11,2,4,1\rangle$,\\
$|\phi_9\rangle=|0,7,3,0\rangle+
|1,11,0,1\rangle+
|2,5,10,1\rangle+
|3,1,4,0\rangle+
|4,6,7,1\rangle+
|5,0,6,0\rangle+
|6,3,5,1\rangle+
|7,9,1,0\rangle+
|8,8,8,1\rangle+
|9,2,11,0\rangle+
|10,4,2,1\rangle+
|11,10,9,0\rangle$,\\
$|\phi_{10}\rangle=|0,1,7,0\rangle+
|1,8,4,0\rangle+
|2,6,1,1\rangle+
|3,0,11,1\rangle+
|4,2,5,0\rangle+
|5,7,8,1\rangle+
|6,11,10,0\rangle+
|7,4,0,1\rangle+
|8,10,2,0\rangle+
|9,9,9,1\rangle+
|10,3,6,0\rangle+
|11,5,3,1\rangle$,\\
$|\phi_{11}\rangle=|0,8,9,1\rangle+
|1,2,8,0\rangle+
|2,9,5,0\rangle+
|3,7,2,1\rangle+
|4,1,6,1\rangle+
|5,3,0,0\rangle+
|6,0,4,1\rangle+
|7,6,11,0\rangle+
|8,5,1,1\rangle+
|9,11,3,0\rangle+
|10,10,10,1\rangle+
|11,4,7,0\rangle$,\\
and $|\phi_{12}\rangle=|0,4,1,0\rangle+
|1,9,10,1\rangle+
|2,3,9,0\rangle+
|3,10,0,0\rangle+
|4,8,3,1\rangle+
|5,2,7,1\rangle+
|6,5,8,0\rangle+
|7,1,5,1\rangle+
|8,7,6,0\rangle+
|9,0,2,1\rangle+
|10,6,4,0\rangle+
|11,11,11,1\rangle$.
\end{example}

\noindent {TABLE IV.} New $m$-QMDS codes $((2d+l+n_1,s^l,d+1))_{s^{2d+l} s_1^1 s_2^1 \cdots s_{n_1}^1}$ and $((2d+l-1+n_1+n_2,s^l,d+1))_{s^{2d+l-1} s_1^1 s_2^1 \cdots s_{n_1}^1 q_1^1q_2^1\cdots q_{n_2}^1}$ obtained by Theorem \ref{tn} (i) with $m=s^l(s_1s_2\cdots s_{n_1}-1)$, $s=12$, $d=1$ and $l=1$.
\vspace{-7mm}
\renewcommand\arraystretch{1}
\setlength{\tabcolsep}{7pt}
\begin{longtable}{lllll}
\label{Table} \\
\hline\hline $m$&$s_1\cdots s_{n_1}$&$((3+n_1,12,2))_{12^3 s_1^1 \cdots s_{n_1}^1}$&$q_1\cdots q_{n_2}$&$((2+n_1+n_2,12,2))_{12^2 s_1^1 \cdots s_{n_1}^1 q_1^1 \cdots q_{n_2}^1}$\\  \hline
\endfirsthead
\multicolumn{5}{c}
{{\bfseries }} \\
\hline\hline $m$&$s_1\cdots s_{n_1}$&$((3+n_1,12,2))_{12^3 s_1^1 \cdots s_{n_1}^1}$&$q_1\cdots q_{n_2}$&$((2+n_1+n_2,12,2))_{12^2 s_1^1 \cdots s_{n_1}^1 q_1^1 \cdots q_{n_2}^1}$\\  \hline
\endhead
\hline \\
\endfoot
\endlastfoot
12&2&$((4,12,2))_{12^3 2^1}$&\tabincell{l}{$6\times2$\\$4\times3$\\$3\times2\times2$}&\tabincell{l}{$((5,12,2))_{12^2 6^1 2^2}$\\$((5,12,2))_{12^2 4^1 3^1 2^1}$\\$((6,12,2))_{12^2 3^1 2^3}$}\\\hline

24&3&$((4,12,2))_{12^3 3^1}$&\tabincell{l}{$6\times2$\\$4\times3$\\$3\times2\times2$}&\tabincell{l}{$((5,12, 2))_{12^2 6^1 3^1 2^1}$\\$((5,12,2))_{12^2 4^1 3^2}$\\$((6,12,2))_{12^2 3^2 2^2}$}\\\hline

36&4&$((4,12,2))_{12^3 4^1}$&\tabincell{l}{$6\times2$\\$4\times3$\\$3\times2\times2$}&\tabincell{l}{$((5,12,2))_{12^2 6^1 4^1 2^1}$\\$((5,12,2))_{12^2 4^2 3^1}$\\$((6,12,2))_{12^2 4^1 3^1 2^2}$}\\\hline
60&6&$((4,12,2))_{12^3 6^1}$&\tabincell{l}{$6\times2$\\$4\times3$\\$3\times2\times2$}&\tabincell{l}{$((5,12,2))_{12^2 6^2 2^1}$\\$((5,12,2))_{12^2 6^1 4^1 3^1}$\\$((6,12,2))_{12^2 6^1 3^1 2^2}$}\\\hline

132&$6\times2$&$((5,12,2))_{12^3 6^1 2^1}$&\tabincell{l}{$6\times2$\\$4\times3$\\$3\times2\times2$\\}&\tabincell{l}{$((6,12,2))_{12^2 6^2 2^2}$\\$((6,12,2))_{12^2 6^1 4^1 3^1 2^1}$\\$((7,12,2))_{12^2 6^1 3^1 2^3}$}\\\hline
132&$4\times3$&$((5,12,2))_{12^3 4^1 3^1}$&\tabincell{l}{$6\times2$\\$4\times3$\\$3\times2\times2$}&\tabincell{l}{$((6,12,2))_{12^2 6^1 4^1 3^1 2^1}$\\$((6,12,2))_{12^2 4^2 3^2}$\\$((7,12,2))_{12^2 4^1 3^2 2^2}$}\\\hline
60&$3\times2$&$((5,12,2))_{12^3 3^1 2^1}$&\tabincell{l}{$6\times2$\\$4\times3$\\$3\times2\times2$}&\tabincell{l}{$((6,12,2))_{12^2 6^1 3^1 2^2}$\\$((6,12,2))_{12^2 4^1 3^2 2^1}$\\$((7,12,2))_{12^2 3^2 2^3}$}\\\hline
36&$2\times2$&$((5,12,2))_{12^3 2^2}$&\tabincell{l}{$6\times2$\\$4\times3$\\$3\times2\times2$}&\tabincell{l}{$((6,12,2))_{12^2 6^1 2^3}$\\$((6,12,2))_{12^2 4^1 3^1 2^2}$\\$((7,12,2))_{12^2 3^1 2^4}$}\\\hline
132&$3\times2\times2$&$((6,12,2))_{12^3 3^1 2^2}$&\tabincell{l}{$6\times2$\\$4\times3$\\$3\times2\times2$}&\tabincell{l}{$((7,12,2))_{12^2 6^1 3^1 2^3}$\\$((7,12,2))_{12^2 4^1 3^2 2^2}$\\$((8,12,2))_{12^2 3^2 2^4}$}\\\hline\hline
\end{longtable}
\vspace{-5mm}
\begin{example}
Take $s=12$, $d=1$, $l=0$ in Theorem \ref{tn} (ii). We can obtain all codes in Table $V$ and their orthogonal bases. For example, the code 1-QMDS code $((5,1,3))_{12^4 2^1}$ has a basis state
$|\phi\rangle=
|0,0,0,0,0\rangle+
|0,1,6,3,1\rangle+
|0,2,8,6,0\rangle+
|0,3,2,1,1\rangle+
|0,4,7,9,0\rangle+
|0,5,1,11,1\rangle+
|0,6,9,2,0\rangle+
|0,7,11,8,1\rangle+
|0,8,4,5,0\rangle+
|0,9,10,4,1\rangle+
|0,10,5,7,1\rangle+
|0,11,3,10,0\rangle+
|1,0,2,6,1\rangle+
|1,1,1,1,0\rangle+
|1,2,7,4,1\rangle+
|1,3,9,7,0\rangle+
|1,4,3,2,1\rangle+
|1,5,8,10,0\rangle+
|1,6,4,11,0\rangle+
|1,7,10,3,0\rangle+
|1,8,6,9,1\rangle+
|1,9,5,0,0\rangle+
|1,10,11,5,1\rangle+
|1,11,0,8,1\rangle+
|2,0,9,11,0\rangle+
|2,1,3,7,1\rangle+
|2,2,2,2,0\rangle+
|2,3,8,5,1\rangle+
|2,4,10,8,0\rangle+
|2,5,4,3,1\rangle+
|2,6,1,9,1\rangle+
|2,7,5,6,0\rangle+
|2,8,11,4,0\rangle+
|2,9,7,10,1\rangle+
|2,10,0,1,0\rangle+
|2,11,6,0,1\rangle+
|3,0,5,4,1\rangle+
|3,1,10,6,0\rangle+
|3,2,4,8,1\rangle+
|3,3,3,3,0\rangle+
|3,4,9,0,1\rangle+
|3,5,11,9,0\rangle+
|3,6,7,1,1\rangle+
|3,7,2,10,1\rangle+
|3,8,0,7,0\rangle+
|3,9,6,5,0\rangle+
|3,10,8,11,1\rangle+
|3,11,1,2,0\rangle+
|4,0,6,10,0\rangle+
|4,1,0,5,1\rangle+
|4,2,11,7,0\rangle+
|4,3,5,9,1\rangle+
|4,4,4,4,0\rangle+
|4,5,10,1,1\rangle+
|4,6,2,3,0\rangle+
|4,7,8,2,1\rangle+
|4,8,3,11,1\rangle+
|4,9,1,8,0\rangle+
|4,10,7,0,0\rangle+
|4,11,9,6,1\rangle+
|5,0,11,2,1\rangle+
|5,1,7,11,0\rangle+
|5,2,1,0,1\rangle+
|5,3,6,8,0\rangle+
|5,4,0,10,1\rangle+
|5,5,5,5,0\rangle+
|5,6,10,7,1\rangle+
|5,7,3,4,0\rangle+
|5,8,9,3,1\rangle+
|5,9,4,6,1\rangle+
|5,10,2,9,0\rangle+
|5,11,8,1,0\rangle+
|6,0,3,8,1\rangle+
|6,1,5,2,0\rangle+
|6,2,10,11,1\rangle+
|6,3,4,10,0\rangle+
|6,4,11,1,0\rangle+
|6,5,9,4,1\rangle+
|6,6,6,6,1\rangle+
|6,7,0,9,0\rangle+
|6,8,2,0,1\rangle+
|6,9,8,7,0\rangle+
|6,10,1,3,1\rangle+
|6,11,7,5,0\rangle+
|7,0,10,5,1\rangle+
|7,1,4,9,1\rangle+
|7,2,0,3,0\rangle+
|7,3,11,6,1\rangle+
|7,4,5,11,0\rangle+
|7,5,6,2,0\rangle+
|7,6,8,0,0\rangle+
|7,7,7,7,1\rangle+
|7,8,1,10,0\rangle+
|7,9,3,1,1\rangle+
|7,10,9,8,0\rangle+
|7,11,2,4,1\rangle+
|8,0,7,3,0\rangle+
|8,1,11,0,1\rangle+
|8,2,5,10,1\rangle+
|8,3,1,4,0\rangle+
|8,4,6,7,1\rangle+
|8,5,0,6,0\rangle+
|8,6,3,5,1\rangle+
|8,7,9,1,0\rangle+
|8,8,8,8,1\rangle+
|8,9,2,11,0\rangle+
|8,10,4,2,1\rangle+
|8,11,10,9,0\rangle+
|9,0,1,7,0\rangle+
|9,1,8,4,0\rangle+
|9,2,6,1,1\rangle+
|9,3,0,11,1\rangle+
|9,4,2,5,0\rangle+
|9,5,7,8,1\rangle+
|9,6,11,10,0\rangle+
|9,7,4,0,1\rangle+
|9,8,10,2,0\rangle+
|9,9,9,9,1\rangle+
|9,10,3,6,0\rangle+
|9,11,5,3,1\rangle+
|10,0,8,9,1\rangle+
|10,1,2,8,0\rangle+
|10,2,9,5,0\rangle+
|10,3,7,2,1\rangle+
|10,4,1,6,1\rangle+
|10,5,3,0,0\rangle+
|10,6,0,4,1\rangle+
|10,7,6,11,0\rangle+
|10,8,5,1,1\rangle+
|10,9,11,3,0\rangle+
|10,10,10,10,1\rangle+
|10,11,4,7,0\rangle+
|11,0,4,1,0\rangle+
|11,1,9,10,1\rangle+
|11,2,3,9,0\rangle+
|11,3,10,0,0\rangle+
|11,4,8,3,1\rangle+
|11,5,2,7,1\rangle+
|11,6,5,8,0\rangle+
|11,7,1,5,1\rangle+
|11,8,7,6,0\rangle+
|11,9,0,2,1\rangle+
|11,10,6,4,0\rangle+
|11,11,11,11,1\rangle$.
\end{example}

\noindent {TABLE V.} New $m_1$-QMDS codes $((2d+n_1,1,d+1))_{s^{2d} s_1^1 s_2^1 \cdots s_{n_1}^1}$ and $m_2$-QMDS codes $((2d-1+n_1+n_2,1,d+1))_{s^{2d-1} s_1^1 s_2^1 \cdots s_{n_1}^1 q_1^1q_2^1\cdots q_{n_2}^1}$ obtained by Theorem \ref{tn} (ii) with $m_1=s_1s_2\cdots s_{n_1}-1$, $m_2=\frac{s s_1s_2\cdots s_{n_1}}{w}-1$ for $w=max\{s_1,s_2,\ldots,s_{n_1},q_1,q_2,\ldots,q_{n_2}\}$, $s=12$ and $d=2$.
\vspace{-7mm}
\renewcommand\arraystretch{1}
\setlength{\tabcolsep}{7pt}
\begin{longtable}{llllll}
\label{Table} \\
\hline\hline $s_1\cdots s_{n_1}$&$m_1$&$((4+n_1,1,3))_{12^4 s_1^1 \cdots s_{n_1}^1}$&$q_1\cdots q_{n_2}$&$m_2$&$((3+n_1+n_2,1,3))_{12^3 s_1^1 \cdots s_{n_1}^1 q_1^1 \cdots q_{n_2}^1}$\\  \hline
\endfirsthead
\multicolumn{6}{c}
{{\bfseries }} \\
\hline\hline $s_1\cdots s_{n_1}$&$m_1$&$((4+n_1,1,3))_{12^4 s_1^1 \cdots s_{n_1}^1}$&$q_1\cdots q_{n_2}$&$m_2$&$((3+n_1+n_2,1,3))_{12^3 s_1^1 \cdots s_{n_1}^1 q_1^1 \cdots q_{n_2}^1}$\\  \hline
\endhead
\hline \\
\endfoot
\endlastfoot
2&1&$((5,1,3))_{12^4 2^1}$&
\tabincell{l}{$6\times2$\\$4\times3$\\$3\times2\times2$}
&\tabincell{l}{3\\5\\7}&\tabincell{l}{$((6,1,3))_{12^3 6^1 2^2}$\\$((6,1,3))_{12^3 4^1 3^1 2^1}$\\$((7,1,3))_{12^3 3^1 2^3}$}\\\hline

3&2&$((5,1,3))_{12^4 3^1}$&
\tabincell{l}{$6\times2$\\$4\times3$\\$3\times2\times2$}
&\tabincell{l}{5\\8\\11}&\tabincell{l}{$((6,1,3))_{12^3 6^1 3^1 2^1}$\\$((6,1,3))_{12^3 4^1 3^2}$\\$((7,1,3))_{12^3 3^2 2^2}$}\\\hline

4&3&$((5,1,3))_{12^4 4^1}$&
\tabincell{l}{$6\times2$\\$4\times3$\\$3\times2\times2$}
&\tabincell{l}{7\\11\\11}&\tabincell{l}{$((6,1,3))_{12^3 6^1 4^1 2^1}$\\$((6,1,3))_{12^3 4^2 3^1}$\\$((7,1,3))_{12^3 4^1 3^1 2^2}$}\\\hline

6&5&$((5,1,3))_{12^4 6^1}$&
\tabincell{l}{$6\times2$\\$4\times3$\\$3\times2\times2$}
&\tabincell{l}{11\\11\\11}&\tabincell{l}{$((6,1,3))_{12^3 6^2 2^1}$\\$((6,1,3))_{12^3 6^1 4^1 3^1}$\\ $((7,1,3))_{12^3 6^1 3^1 2^2}$}\\\hline

$6\times2$&11&$((6,1,3))_{12^4 6^1 2^1}$&
\tabincell{l}{$6\times2$\\$4\times3$\\$3\times2\times2$}
&\tabincell{l}{23\\23\\23}&\tabincell{l}{$((7,1,3))_{12^3 6^2 2^2}$\\$((7,1,3))_{12^3 6^1 4^1 3^1 2^1}$\\$((8,1,3))_{12^3 6^1 3^1 2^3}$}\\\hline

$4\times3$&11&$((6,1,3))_{12^4 4^1 3^1}$&
\tabincell{l}{$6\times2$\\$4\times3$\\$3\times2\times2$}
&\tabincell{l}{23\\35\\35}&\tabincell{l}{$((7,1,3))_{12^3 6^1 4^1 3^1 2^1}$\\$((7,1,3))_{12^3 4^2 3^2}$\\$((8,1,3))_{12^3 4^1 3^2 2^2}$}\\\hline

$3\times2$&5&$((6,1,3))_{12^4 3^1 2^1}$&
\tabincell{l}{$6\times2$\\$4\times3$\\$3\times2\times2$}
&\tabincell{l}{11\\17\\23}&\tabincell{l}{$((7,1,3))_{12^3 6^1 3^1 2^2}$\\$((7,1,3))_{12^3 4^1 3^2 2^1}$\\$((8,1,3))_{12^3 3^2 2^3}$}\\\hline

$2\times2$&3&$((6,1,3))_{12^4 2^2}$&
\tabincell{l}{$6\times2$\\$4\times3$\\$3\times2\times2$}
&\tabincell{l}{7\\11\\15}&\tabincell{l}{$((7,1,3))_{12^3 6^1 2^3}$\\$((7,1,3))_{12^3 4^1 3^1 2^2}$\\$((8,1,3))_{12^3 3^1 2^4}$}\\\hline

$3\times2\times2$&11&$((7,1,3))_{12^4 3^1 2^2}$&
\tabincell{l}{$6\times2$\\$4\times3$\\$3\times2\times2$}
&\tabincell{l}{23\\35\\47}&\tabincell{l}{$((8,1,3))_{12^3 6^1 3^1 2^3}$\\$((8,1,3))_{12^3 4^1 3^2 2^2}$\\$((9,1,3))_{12^3 3^2 2^4}$}\\\hline\hline
\end{longtable}
\vspace{-5mm}
The following two corollaries follow immediately from Theorem \ref{tn}.
\begin{corollary}\label{uv}
Suppose $d\geq 1$ and $s_1,s_2,\ldots,s_{n_1},q_1,q_2,\ldots,q_{n_2}\geq2$ are all integers and that $s=u_1u_2\cdots u_v$ for prime powers $u_1,u_2,\ldots,u_v$ and $min\{u_1,u_2,\ldots,u_v\}\geq 2d+2l$. If $(s_1s_2\cdots s_{n_1})|s$ and $s=q_1 q_2 \cdots q_{n_2}$, then

(i) when $l\geq1$, there exist two $(s_1s_2\cdots s_{n_1}-1)s^l$-QMDS codes $((2d+l+n_1,s^l,d+1))_{s^{2d+l} s_1^1 s_2^1\cdots s_{n_1}^1}$ and $((2d+l-1+n_1+n_2,s^l,d+1))_{s^{2d+l-1} s_1^1 s_2^1 \cdots s_{n_1}^1 q_1^1 q_2^1\cdots q_{n_2}^1}$.

(ii) when $l=0$, there exist an $(s_1s_2\cdots s_{n_1}-1)$-QMDS code $((2d+n_1,1,d+1))_{s^{2d} s_1^1 s_2^1 \cdots s_{n_1}^1}$ and an $(\frac{ss_1s_2\cdots s_{n_1}}{w}-1)$-QMDS code  $((2d-1+n_1+n_2,1,d+1))_{s^{2d-1} s_1^1 s_2^1 \cdots s_{n_1}^1 q_1^1q_2^1\cdots q_{n_2}^1}$ for $w=max\{s_1,s_2,\ldots,s_{n_1},q_1,q_2,\ldots,q_{n_2}\}$.
\end{corollary}

\begin{proof}
By Lemma \ref{bush}, $OA(u_1^{d+l},2d+2l+1,u_1,d+l)$, $\ldots$, $OA(u_{v-1}^{d+l},2d+2l+1,u_{v-1},d+l)$ and $OA(u_v^{d+l},2d+2l+1,u_v,d+l)$ exist. By Lemma \ref{chengtui}, an $OA(s^{d+l},2d+2l+1,s,d+l)$ exists. It follows from Theorem \ref{tn} that the corollary holds.
\end{proof}

\begin{example}
Take $s=16$, $d=4$, $l=1$ in Corollary \ref{uv} (i). We can obtain
five 16-QMDS codes $((10,16,5))_{16^9 2^1}$, $((11,16,5))_{16^8 8^1 2^2}$, $((11,16,5))_{16^8 4^2 2^1}$, $((12,16,5))_{16^8 4^1 2^3}$, $((13,16,5))_{16^8 2^5}$,
ten 48-QMDS codes $((10,16,5))_{16^9 4^1}$, $((11,16,5))_{16^8 8^1 4^1 2^1}$, $((11,16,5))_{16^8 4^3}$, $((12,16,5))_{16^8 4^2 2^2}$, $((13,16,5))_{16^8 4^1 2^4}$,
$((11,16,5))_{16^9 2^2}$, $((12,16,5))_{16^8 8^1 2^2}$, $((12,16,5))_{16^8 4^2  2^2}$, $((13,16,5))_{16^8 4^1 2^4}$, $((14,16,5))_{16^8 2^6}$,
fifteen 112-QMDS codes $((10,16,5))_{16^9 8^1}$, $((11,16,5))_{16^8 8^2 2^1}$, $((11,16,5))_{16^8 8^1 4^2}$, $((12,16,5))_{16^8 8^1 4^1 2^2}$, $((13,16,5))_{16^8 8^1 2^4}$,
$((11,16,5))_{16^9 4^1 2^1}$, $((12,16,5))_{16^8 8^1 4^1 2^2}$, $((12,16,5))_{16^8 4^3 2^1}$, $((13,16,5))_{16^8 4^2 2^3}$, $((14,16,5))_{16^8 4^1 2^5}$,
$((12,16,5))_{16^9 2^3}$, $((13,16,5))_{16^8 8^1 2^4}$, $((13,16,5))_{16^8 4^2 2^3}$, $((14,16,5))_{16^8 4^1 2^5}$, $((15,16,5))_{16^8 2^7}$
and twenty 240-QMDS codes $((11,16,5))_{16^9 8^1 2^1}$, $((12,16,5))_{16^8 8^2 2^2}$, $((12,16,5))_{16^8 8^1 4^2 2^1}$, $((13,16,5))_{16^8 8^1 4^1 2^3}$, $((14,16,5))_{16^8 8^1 2^5}$,
$((11,16,5))_{16^9 4^2}$, $((12,16,5))_{16^8 8^1 4^2 2^1}$, $((12,16,5))_{16^8 4^4}$, $((13,16,5))_{16^8 4^3 2^2}$,
$((14,16,5))_{16^8 4^2 2^4}$,
$((12,16,5))_{16^9 4^1 2^2}$, $((13,16,5))_{16^8 8^1 4^1 2^3}$, $((13,16,5))_{16^8 4^3 2^2}$, $((14,16,5))_{16^8 4^2 2^4}$, $((15,16,5))_{16^8 4^1 2^6}$,
$((13,16,5))_{16^9 2^4}$, $((14,16,5))_{16^8 8^1 2^5}$, $((14,16,5))_{16^8 4^2 2^4}$, $((15,16,5))_{16^8 4^1 2^6}$, $((16,16,5))_{16^8 2^8}$.
\end{example}

\begin{example}
Take $s=16$, $d=4$, $l=0$ in Corollary \ref{uv} (ii). We have a 1-QMDS codes $((9,1,5))_{16^8 2^1}$,
three 3-QMDS codes $((10,1,5))_{16^7 8^1 2^2}$, $((9,1,5))_{16^8 4^1}$, $((10,1,5))_{16^8 2^2}$,
seven 7-QMDS codes $((10,1,5))_{16^7 4^2 2^1}$, $((11,1,5))_{16^7 4^1 2^3}$, $((10,1,5))_{16^7 8^1 4^1 2^1}$, $((9,1,5))_{16^8 8^1}$, $((10,1,5))_{16^8 4^1 2^1}$, $((11,1,5))_{16^7 8^1 2^3}$, $((11,1,5))_{16^8 2^3}$,
sixteen 15-QMDS codes $((12,1,5))_{16^7 2^5}$, $((10,1,5))_{16^7 4^3}$, $((11,1,5))_{16^7 4^2 2^2}$, $((12,1,5))_{16^7 4^1 2^4}$, $((10,1,5))_{16^7 8^2 2^1}$, $((10,1,5))_{16^7 8^1 4^2}$, $((11,1,5))_{16^7 8^1 4^1 2^2}$, $((12,1,5))_{16^7 8^1 2^4}$, $((10,1,5))_{16^8 8^1 2^1}$, $((10,1,5))_{16^8 4^2}$, $((11,1,5))_{16^7 8^1 4^1 2^2}$, $((11,1,5))_{16^7 4^2 2^2}$, $((12,1,5))_{16^7 4^1 2^4}$, $((11,1,5))_{16^8 4^1 2^2}$, $((12,1,5))_{16^7 8^1 2^4}$, $((12,1,5))_{16^8 2^4}$,
thirteen 31-QMDS codes $((11,1,5))_{16^7 8^2 2^2}$, $((11,1,5))_{16^7 8^1 4^2 2^1}$, $((12,1,5))_{16^7 8^1 4^1 2^3}$, $((13,1,5))_{16^7 8^1 2^5}$, $((11,1,5))_{16^7 8^1 4^2 2^1}$, $((11,1,5))_{16^7 4^3 2^1}$, $((12,1,5))_{16^7 4^2 2^3}$, $((13,1,5))_{16^7 4^1 2^5}$, $((13,1,5))_{16^7 2^6}$, $((12,1,5))_{16^7 8^1 4^1 2^3}$, $((12,1,5))_{16^7 4^2 2^3}$, $((13,1,5))_{16^7 4^1 2^5}$, $((13,1,5))_{16^7 8^1 2^5}$,
nine 63-QMDS codes $((11,1,5))_{16^7 4^4}$, $((12,1,5))_{16^7 4^3 2^2}$, $((13,1,5))_{16^7 4^2 2^4}$, $((12,1,5))_{16^7 4^3 2^2}$, $((13,1,5))_{16^7 4^2 2^4}$, $((14,1,5))_{16^7 4^1 2^6}$, $((14,1,5))_{16^7 2^7}$, $((13,1,5))_{16^7 4^2 2^4}$, $((14,1,5))_{16^7 4^1 2^6}$
and a 127-QMDS code $((15,1,5))_{16^7 2^8}$.
\end{example}

\begin{example}
Take $s=56$, $d=2$ and $l=1$ in Corollary \ref{uv} (i). We can obtain the following 133 m-QMDS codes.

(1) Seven 56-QMDS codes $((6,56,3))_{56^5 2^1}$, $((7,56,3))_{56^4 28^1 2^2}$, $((7,56,3))_{56^4 14^1 4^1 2^1}$, $((7,56,3))_{56^4 8^1 7^1 2^1}$, $((8,56,3))_{56^4 14^1 2^3}$, $((8,56,3))_{56^4 7^1 4^1 2^2}$, $((9,56,3))_{56^4 7^1 2^4}$.

(2) Fourteen 168-QMDS codes
$((6,56,3))_{56^5 4^1}$, $((7,56,3))_{56^4 28^1 4^1 2^1}$,
$((7,56,3))_{56^4 14^1 4^2}$, $((7,56,3))_{56^4 8^1 7^1 4^1}$,  $((8,56,3))_{56^4 14^1 4^1 2^2}$, $((8,56,3))_{56^4 7^1 4^2 2^1}$, $((9,56,3))_{56^4 7^1 4^1 2^3}$, $((7,56,3))_{56^5 2^2}$,
$((8,56,3))_{56^4 28^1 2^3}$, $((8,56,3))_{56^4 14^1 4^1 2^2}$, $((8,56,3))_{56^4 8^1 7^1 2^2}$, $((9,56,3))_{56^4 14^1 2^4}$, $((9,56,3))_{56^4 7^1 4^1 2^3}$, $((10,56,3))_{56^4 7^1 2^5}$.

(3) Seven 336-QMDS codes
$((6,56,3))_{56^5 7^1}$, $((7,56,3))_{56^4 28^1 7^1 2^1}$,
$((7,56,3))_{56^4 14^1 7^1 4^1}$, $((7,56,3))_{56^4 8^1 7^2}$, $((8,56,3))_{56^4 14^1 7^2 2^2}$, $((8,56,3))_{56^4 7^2 4^1 2^1}$, $((9,56,3))_{56^4 7^2 2^3}$.

(4) Twenty-one 392-QMDS codes
$((6,56,3))_{56^5 8^1}$, $((7,56,3))_{56^4 28^1 8^1 2^1}$,
$((7,56,3))_{56^4 14^1 8^1 4^1}$, $((7,56,3))_{56^4 8^2 7^1}$, $((8,56,3))_{56^4 14^1 8^1 2^2}$, $((8,56,3))_{56^4 8^1 7^1 4^1 2^1}$, $((9,56,3))_{56^4 8^1 7^1 2^3}$, $((7,56,3))_{56^5 4^1 2^1}$, $((8,56,3))_{56^4 28^1 4^1 2^2}$, $((8,56,3))_{56^4 14^1 4^2 2^1}$, $((8,56,3))_{56^4 8^1 7^1 4^1 2^1}$, $((9,56,3))_{56^4 14^1 4^1 2^3}$, $((9,56,3))_{56^4 7^1 4^2 2^2}$, $((10,56,3))_{56^4 7^1 4^1 2^4}$,
$((8,56,3))_{56^5 2^3}$, $((9,56,3))_{56^4 28^1 2^4}$,
$((9,56,3))_{56^4 14^1 4^1 2^3}$, $((9,56,3))_{56^4 8^1 7^1 2^3}$, $((10,56,3))_{56^4 14^1 2^5}$, $((10,56,3))_{56^4 7^1 4^1 2^4}$,  $((11,56,3))_{56^4 7^1 2^6}$.

(5) Fourteen 728-QMDS codes
$((6,56,3))_{56^5 14^1}$, $((7,56,3))_{56^4 28^1 14^1 2^1}$, $((7,56,3))_{56^4 14^2 4^1}$, $((7,56,3))_{56^4 14^1 8^1 7^1}$, $((8,56,3))_{56^4 14^2 2^2}$,  $((8,56,3))_{56^4 14^1 7^1 4^1 2^1}$, $((9,56,3))_{56^4 14^1 7^1 2^3}$,
$((7,56,3))_{56^5 7^1 2^1}$, $((8,56,3))_{56^4 28^1 7^1 2^2}$, $((8,56,3))_{56^4 14^1 7^1 4^1 2^1}$, $((8,56,3))_{56^4 8^1 7^2 2^1}$, $((9,56,3))_{56^4 14^1 7^1 2^3}$, $((9,56,3))_{56^4 7^2 4^1 2^2}$, $((10,56,3))_{56^4 7^22^4}$.

(6) Twenty-eight 1512-QMDS code
$((6,56,3))_{56^5 28^1}$, $((7,56,3))_{56^4 28^2 2^1}$,
$((7,56,3))_{56^4 28^1 14^1 4^1}$, $((7,56,3))_{56^4 28^1 8^1 7^1}$, $((8,56,3))_{56^4 28^1 14^1 2^2}$,
$((8,56,3))_{56^4 28^1 7^1 4^1 2^1}$, $((9,56,3))_{56^4 28^1 7^1 2^3}$,
$((7,56,3))_{56^5 14^1 2^1}$, $((8,56,3))_{56^4 28^1 14^1 2^2}$, $((8,56,3))_{56^4 14^2 4^1 2^1}$,
$((8,56,3))_{56^4 14^1 8^1 7^1 2^1}$, $((9,56,3))_{56^4 14^2 2^3}$, $((9,56,3))_{56^4 14^1 7^1 4^1 2^2}$, $((10,56,3))_{56^4 14^1 7^1 2^4}$,
$((7,56,3))_{56^5 7^1 4^1}$,
$((8,56,3))_{56^4 28^1 7^1 4^1 2^1}$, $((8,56,3))_{56^4 14^1 7^1 4^2}$, $((8,56,3))_{56^4 8^1 7^2 4^1}$, $((9,56,3))_{56^4 14^1 7^1 4^1 2^2}$, $((9,56,3))_{56^4 7^2 4^2 2^1}$,
$((10,56,3))_{56^4 7^2 4^1 2^3}$, $((8,56,3))_{56^5 7^1 2^2}$, $((9,56,3))_{56^4 28^1 7^1 2^3}$, $((9,56,3))_{56^4 14^1 7^1 4^1 2^2}$, $((9,56,3))_{56^4 8^1 7^2 2^2}$, $((10,56,3))_{56^4 14^1 7^1 2^4}$, $((10,56,3))_{56^4 7^2 4^1 2^3}$, $((11,56,3))_{56^4 7^2 2^5}$.

(7) Forty-two 3080-QMDS codes $((7,56,3))_{56^5 28^1 2^1}$, $((8,56,3))_{56^4 28^2 2^2}$, $((8,56,3))_{56^4 28^1 14^1 4^1 2^1}$, $((8,56,3))_{56^4 28^1 8^1 7^1 2^1}$, $((9,56,3))_{56^4 28^1 14^1  2^3}$, $((9,56,3))_{56^4 28^1 7^1 4^1 2^2}$, $((10,56,3))_{56^4 28^1 7^1 2^4}$,
$((7,56,3))_{56^5 14^1 4^1}$, $((8,56,3))_{56^4 28^1 14^1 4^1 2^1}$, $((8,56,3))_{56^4 14^2 4^2}$,
$((8,56,3))_{56^4 14^1 8^1 7^1 4^1}$, $((9,56,3))_{56^4 14^2 4^1 2^2}$, $((9,56,3))_{56^4 14^1 7^1 4^2 2^1}$, $((10,56,3))_{56^4 14^1 7^1 4^1 2^3}$,
$((7,56,3))_{56^5 8^1 7^1}$,  $((8,56,3))_{56^4 28^1 8^1 7^1 2^1}$, $((8,56,3))_{56^4 14^1 8^1 7^1 4^1}$, $((8,56,3))_{56^4 8^2 7^2}$, $((9,56,3))_{56^4 14^1 8^1 7^1 2^2}$, $((9,56,3))_{56^4 8^1 7^2 4^1 2^1}$,  $((10,56,3))_{56^4 8^1 7^2 2^3}$,
$((8,56,3))_{56^5 14^1 2^2}$, $((9,56,3))_{56^4 28^1 14^1 2^3}$, $((9,56,3))_{56^4 14^2 4^1 2^2}$, $((9,56,3))_{56^4 14^1 8^1 7^1 2^2}$,  $((10,56,3))_{56^4 14^2 2^4}$, $((10,56,3))_{56^4 14^1 7^1 4^1 2^3}$, $((11,56,3))_{56^4 14^1 7^1 2^5}$,
$((8,56,3))_{56^5 7^1 4^1 2^1}$, $((9,56,3))_{56^4 28^1 7^1 4^1 2^2}$,  $((9,56,3))_{56^4 14^1 7^1 4^2 2^1}$, $((9,56,3))_{56^4 8^1 7^2 4^1 2^1}$, $((10,56,3))_{56^4 14^1 7^1 4^1 2^3}$, $((10,56,3))_{56^4 7^2 4^2 2^2}$,  $((11,56,3))_{56^4 7^2 4^1 2^4}$,
$((9,56,3))_{56^5 7^1 2^3}$, $((10,56,3))_{56^4 28^1 7^1 2^4}$, $((10,56,3))_{56^4 14^1 7^1 4^1 2^3}$, $((10,56,3))_{56^4 8^1 7^2 2^3}$,  $((11,56,3))_{56^4 14^1 7^1 2^5}$, $((11,56,3))_{56^4 7^2 4^1 2^4}$, $((12,56,3))_{56^4 7^2 2^6}$.
\end{example}

\noindent {TABLE VI.} New $m$-QMDS codes $((2d+l+n_1,s^l,d+1))_{s^{2d+l} s_1^1  s_2^1\cdots s_{n_1}^1}$ with $m=(s_1 \cdots s_{n_1}-1)s^l$ from Corollary \ref{uv}.
\vspace{-7mm}
\renewcommand\arraystretch{1}
\setlength{\tabcolsep}{3.4pt}
\begin{longtable}{cccccc}
\label{Table} \\
\hline\hline $l$&$d$&$s$&$s_1s_2\cdots s_{n_1}$&$((2d+l+n_1,s^l,d+1))_{s^{2d+l} s_1^1 \cdots s_{n_1}^1}$&$m$ \\  \hline
\endfirsthead
\multicolumn{6}{c}
{{\bfseries }} \\
\hline\hline $l$&$d$&$s$&$s_1s_2\cdots s_{n_1}$&$((2d+l+n_1,s^l,d+1))_{s^{2d+l} s_1^1 \cdots s_{n_1}^1}$&$m$ \\  \hline
\endhead
\hline \\
\endfoot
\endlastfoot
0&1&4&$2\times2$&$((4,1,2))_{4^2 2^2}$&3\\

0&1&6&$3$&$((3,1,2))_{6^2 3^1}$&2\\
0&1&6&$2\times3$&$((4,1,2))_{6^2 2^1 3^1}$&5\\

0&1&8&$4$&$((3,1,2))_{8^2 4^1}$&3\\
0&1&8&$4\times2$&$((4,1,2))_{8^2 4^1 2^1}$&7\\
0&1&8&$2\times2$&$((4,1,2))_{8^2 2^2}$&3\\
0&1&8&$2\times 2\times2$&$((5,1,2))_{8^2 2^3}$&7\\
0&1&$s\geq9$&$s_1 s_2 \cdots s_{n_1}$&$((2+n_1,1,2))_{s^2 s_1^1 s_2^1 \cdots s_{n_1}^1}$&$s_1s_2\cdots s_{n_1}-1$\\

0&2&4&$2\times2$&$((6,1,3))_{4^4 2^2}$&3\\

0&2&8&$4$&$((5,1,3))_{8^4 4^1}$&3\\
0&2&8&$4\times2$&$((6,1,3))_{8^4 4^1 2^1}$&7\\
0&2&8&$2\times2$&$((6,1,3))_{8^4 2^2}$&3\\
0&2&8&$2\times 2\times2$&$((7,1,3))_{8^4 2^3}$&7\\
0&2&$s\geq9$&$s_1 s_2 \cdots s_{n_1}$&$((4+n_1,1,3))_{s^4 s_1^1 s_2^1 \cdots s_{n_1}^1}$&$s_1s_2\cdots s_{n_1}-1$\\

0&3&8&$4$&$((7,1,4))_{8^6 4^1}$&3\\
0&3&8&$4\times2$&$((8,1,4))_{8^6 4^1 2^1}$&7\\
0&3&8&$2\times2$&$((8,1,4))_{8^6 2^2}$&3\\
0&3&8&$2\times 2\times2$&$((9,1,4))_{8^6 2^3}$&7\\

0&3&9&$3$&$((7,1,4))_{9^6 3^1}$&2\\
0&3&9&$3\times 3$&$((8,1,4))_{9^6 3^2}$&8\\
0&3&$s\geq16$&$s_1 s_2 \cdots s_{n_1}$&$((6+n_1,1,4))_{s^6 s_1^1 s_2^1 \cdots s_{n_1}^1}$&$s_1s_2\cdots s_{n_1}-1$\\

0&4&8&$4$&$((9,1,5))_{8^8 4^1}$&3\\
0&4&8&$4\times2$&$((10,1,5))_{8^8 4^1 2^1}$&7\\
0&4&8&$2\times2$&$((10,1,5))_{8^8 2^2}$&3\\
0&4&8&$2\times 2\times2$&$((11,1,5))_{8^8 2^3}$&7\\

0&4&9&$3$&$((9,1,5))_{9^8 3^1}$&2\\
0&4&9&$3\times 3$&$((10,1,5))_{9^8 3^2}$&8\\
0&4&$s\geq16$&$s_1 s_2 \cdots s_{n_1}$&$((8+n_1,1,5))_{s^8 s_1^1 s_2^1 \cdots s_{n_1}^1}$&$s_1s_2\cdots s_{n_1}-1$\\

0&$d\geq 5$&$s\geq16$&$s_1 s_2 \cdots s_{n_1}$&$((2d+n_1,1,d+1))_{s^{2d} s_1^1 s_2^1 \cdots s_{n_1}^1}$&$s_1s_2\cdots s_{n_1}-1$\\

1&1&4&$2$&$((4,4,2))_{4^3 2^1}$&4\\
1&1&4&$2\times2$&$((5,4,2))_{4^3 2^2}$&12\\

1&1&8&$2$&$((4,8,2))_{8^3 2^1}$&8\\
1&1&8&$4$&$((4,8,2))_{8^3 4^1}$&24\\
1&1&8&$4\times2$&$((5,8,2))_{8^3 4^1 2^1}$&56\\
1&1&8&$2\times2$&$((5,8,2))_{8^3 2^2}$&24\\
1&1&8&$2\times 2\times2$&$((6,8,2))_{8^3 2^3}$&56\\
1&1&$s\geq9$&$s_1 s_2 \cdots s_{n_1}$&$((3+n_1,s,2))_{s^3 s_1^1 s_2^1 \cdots s_{n_1}^1}$&$(s_1s_2\cdots s_{n_1}-1)s$\\

1&2&4&$2$&$((6,4,3))_{4^4 2^1}$&4\\
1&2&4&$2\times2$&$((7,4,3))_{4^4 2^2}$&12\\

1&2&8&$2$&$((6,8,3))_{8^4 2^1}$&4\\
1&2&8&$4$&$((6,8,3))_{8^4 4^1}$&12\\
1&2&8&$4\times2$&$((7,8,3))_{8^4 4^1 2^1}$&56\\
1&2&8&$2\times2$&$((7,8,3))_{8^4 2^2}$&24\\
1&2&8&$2\times 2\times2$&$((8,8,3))_{8^4 2^3}$&56\\
1&2&$s\geq9$&$s_1 s_2 \cdots s_{n_1}$&$((5+n_1,s,2))_{s^4 s_1^1 s_2^1 \cdots s_{n_1}^1}$&$(s_1s_2\cdots s_{n_1}-1)s$\\

1&3&8&$2$&$((8,8,4))_{8^7 2^1}$&8\\
1&3&8&$4$&$((8,8,4))_{8^7 4^1}$&24\\
1&3&8&$4\times2$&$((9,8,4))_{8^7 4^1 2^1}$&56\\
1&3&8&$2\times2$&$((9,8,4))_{8^7 2^2}$&24\\
1&3&8&$2\times 2\times2$&$((10,8,4))_{8^7 2^3}$&56\\

1&3&9&$3$&$((8,9,4))_{9^7 3^1}$&18\\
1&3&9&$3\times 3$&$((9,9,4))_{9^7 3^2}$&72\\
1&3&$s\geq16$&$s_1 s_2 \cdots s_{n_1}$&$((7+n_1,s,4))_{s^7 s_1^1 s_2^1 \cdots s_{n_1}^1}$&$(s_1s_2\cdots s_{n_1}-1)s$\\

1&4&$16$&2&$((10,16,5))_{16^9 2^1}$&16\\
1&4&$16$&4&$((10,16,5))_{16^9 2^1}$&48\\
1&4&$16$&8&$((10,16,5))_{16^9 2^1}$&112\\
1&4&$16$&$8\times2$&$((11,16,5))_{16^9 2^1}$&240\\
1&4&$16$&$4\times4$&$((11,16,5))_{16^9 2^1}$&240\\
1&4&$16$&$4\times2$&$((11,16,5))_{16^9 2^1}$&112\\
1&4&$16$&$2\times2$&$((11,16,5))_{16^9 2^1}$&48\\
1&4&$16$&$4\times2\times2$&$((12,16,5))_{16^9 2^1}$&240\\
1&4&$16$&$2\times2\times2$&$((12,16,5))_{16^9 2^1}$&112\\
1&4&$16$&$2\times2\times2\times2$&$((13,16,5))_{16^9 2^1}$&240\\
1&4&$s\geq27$&$s_1 s_2 \cdots s_{n_1}$&$((9+n_1,s,5))_{s^7 s_1^1 s_2^1 \cdots s_{n_1}^1}$&$(s_1s_2\cdots s_{n_1}-1)s$\\

1&$d\geq 5$&$s\geq16$&$s_1 s_2 \cdots s_{n_1}$&$((2d+1+n_1,s,d+1))_{s^{2d+1} s_1^1 s_2^1 \cdots s_{n_1}^1}$&$(s_1s_2\cdots s_{n_1}-1)s$\\

$l\geq2$&$d\geq 1$&$s\geq8$&$s_1 s_2 \cdots s_{n_1}$&$((2d+l+n_1,s^l,d+1))_{s^{2d+l} s_1^1 s_2^1 \cdots s_{n_1}^1}$&$(s_1s_2\cdots s_{n_1}-1)s^l$\\\hline\hline
\end{longtable}
\vspace{-5mm}
\begin{corollary}\label{5lie}
An $(s_1s_2\cdots s_n-1)$-QMDS code $((4+n,1,3))_{s^4 s_1^1 s_2^1 \cdots s_n^1}$ exists for $(s_1s_2\cdots s_n)|s$, $s\geq4$ and $s\neq 6,10$.
\end{corollary}

\begin{proof}
It is similar to Corollary \ref{1wei} (ii). Since there exists an $OA(s^2,5,s,2)$ for $s\geq4$ and $s\neq 6,10$,
it follows from Theorem \ref{tn} that the corollary holds.
\end{proof}

\begin{example}
Take $s=9$ and $s_1=3$. Corollary \ref{5lie} produces a 2-QMDS code $((5,1,3))_{9^4 3^1}$ with a basis state
$|\phi\rangle=
|0	0	0	0	0\rangle+
|0	1	2	3	1\rangle+
|0	2	1	6	2\rangle+
|0	3	6	7	0\rangle+
|0	4	8	1	1\rangle+
|0	5	7	4	2\rangle+
|0	6	3	5	0\rangle+
|0	7	5	8	1\rangle+
|0	8	4	2	2\rangle+
|1	0	2	7	2\rangle+
|1	1	1	1	0\rangle+
|1	2	0	4	1\rangle+
|1	3	8	5	2\rangle+
|1	4	7	8	0\rangle+
|1	5	6	2	1\rangle+
|1	6	5	0	2\rangle+
|1	7	4	3	0\rangle+
|1	8	3	6	1\rangle+
|2	0	1	5	1\rangle+
|2	1	0	8	2\rangle+
|2	2	2	2	0\rangle+
|2	3	7	0	1\rangle+
|2	4	6	3	2\rangle+
|2	5	8	6	0\rangle+
|2	6	4	7	1\rangle+
|2	7	3	1	2\rangle+
|2	8	5	4	0\rangle+
|3	0	6	8	1\rangle+
|3	1	8	2	2\rangle+
|3	2	7	5	0\rangle+
|3	3	3	3	1\rangle+
|3	4	5	6	2\rangle+
|3	5	4	0	0\rangle+
|3	6	0	1	1\rangle+
|3	7	2	4	2\rangle+
|3	8	1	7	0\rangle+
|4	0	8	3	0\rangle+
|4	1	7	6	1\rangle+
|4	2	6	0	2\rangle+
|4	3	5	1	0\rangle+
|4	4	4	4	1\rangle+
|4	5	3	7	2\rangle+
|4	6	2	8	0\rangle+
|4	7	1	2	1\rangle+
|4	8	0	5	2\rangle+
|5	0	7	1	2\rangle+
|5	1	6	4	0\rangle+
|5	2	8	7	1\rangle+
|5	3	4	8	2\rangle+
|5	4	3	2	0\rangle+
|5	5	5	5	1\rangle+
|5	6	1	3	2\rangle+
|5	7	0	6	0\rangle+
|5	8	2	0	1\rangle+
|6	0	3	4	2\rangle+
|6	1	5	7	0\rangle+
|6	2	4	1	1\rangle+
|6	3	0	2	2\rangle+
|6	4	2	5	0\rangle+
|6	5	1	8	1\rangle+
|6	6	6	6	2\rangle+
|6	7	8	0	0\rangle+
|6	8	7	3	1\rangle+
|7	0	5	2	1\rangle+
|7	1	4	5	2\rangle+
|7	2	3	8	0\rangle+
|7	3	2	6	1\rangle+
|7	4	1	0	2\rangle+
|7	5	0	3	0\rangle+
|7	6	8	4	1\rangle+
|7	7	7	7	2\rangle+
|7	8	6	1	0\rangle+
|8	0	4	6	0\rangle+
|8	1	3	0	1\rangle+
|8	2	5	3	2\rangle+
|8	3	1	4	0\rangle+
|8	4	0	7	1\rangle+
|8	5	2	1	2\rangle+
|8	6	7	2	0\rangle+
|8	7	6	5	1\rangle+
|8	8	8	8	2\rangle$.
\end{example}

\noindent {TABLE VII.} New 2-QMDS codes $((2d+1,1,d+1))_{s^{2d} 3^1}$ from Corollaries \ref{uv} and \ref{5lie}.
\vspace{-7mm}
\renewcommand\arraystretch{1}
\setlength{\tabcolsep}{7pt}
\begin{longtable}{lll}
\label{Table} \\
\hline\hline $d$&$s$&$((2d+1,1,d+1))_{s^{2d} 3^1}$\\  \hline
\endfirsthead
\multicolumn{3}{c}
{{\bfseries }} \\
\hline\hline $d$&$s$&$((2d+1,1,d+1))_{s^{2d} 3^1}$\\  \hline
\endhead
\hline \\
\endfoot
\endlastfoot
1&6&$((3,1,2))_{6^2 3^1}$\\
1&9&$((3,1,2))_{9^2 3^1}$\\
1&12&$((3,1,2))_{12^2 3^1}$\\
1&15&$((3,1,2))_{15^2 3^1}$\\
1&18&$((3,1,2))_{18^2 3^1}$\\
1&21&$((3,1,2))_{21^2 3^1}$\\
1&$s=3k(k\geq8)$&$((3,1,2))_{s^2 3^1}$\\\hline

2&9&$((5,1,3))_{9^4 3^1}$\\
2&12&$((5,1,3))_{12^4 3^1}$\\
2&15&$((5,1,3))_{15^4 3^1}$\\
2&18&$((5,1,3))_{18^4 3^1}$\\
2&21&$((5,1,3))_{21^4 3^1}$\\
2&$s=3k(k\geq8)$&$((5,1,3))_{s^4 3^1}$\\\hline

3&9&$((7,1,4))_{9^6 3^1}$\\
3&27&$((7,1,4))_{27^6 3^1}$\\
3&45&$((7,1,4))_{45^6 3^1}$\\
3&63&$((7,1,4))_{63^6 3^1}$\\
3&72&$((7,1,4))_{72^6 3^1}$\\
3&81&$((7,1,4))_{81^6 3^1}$\\
3&\tabincell{l}{$s=9\times 3^p \times u_1\times u_2\times \cdots \times u_v$\\($p\geq0$, $u_i$ is a prime power and $u_i\geq 5$)}&$((7,1,4))_{s^6 3^1}$\\\hline

4&9&$((9,1,5))_{9^8 3^1}$\\
4&27&$((9,1,5))_{27^8 3^1}$\\
4&63&$((9,1,5))_{63^8 3^1}$\\
4&72&$((9,1,5))_{72^8 3^1}$\\
4&81&$((9,1,5))_{81^8 3^1}$\\
4&\tabincell{l}{$s=9\times 3^p \times u_1\times u_2\times \cdots \times u_v$\\($p\geq0$, $u_i$ is a prime power and $u_i\geq 8$)}&$((9,1,5))_{s^8 3^1}$\\\hline

5&27&$((11,1,6))_{27^{10} 3^1}$\\
5&81&$((11,1,6))_{81^{10} 3^1}$\\
5&243&$((11,1,6))_{243^{10} 3^1}$\\
5&297&$((11,1,6))_{297^{10} 3^1}$\\
5&\tabincell{l}{$s=27\times 3^p \times u_1\times u_2\times \cdots \times u_v$\\($p\geq0$, $u_i$ is a prime power and $u_i\geq 11$)}&$((11,1,6))_{s^{10} 3^1}$\\\hline

$\cdots$&$\cdots$&$\cdots$\\\hline\hline
\end{longtable}
By using the expansive replacement method, we can obtain a lot of $m$-QMDS codes from known QMDS codes.

\begin{theorem}\label{huan}
If a QMDS code $((2d+n,s_2^n,d+1))_{s_1^{2d} s_2^{n}}$ is constructed by Lemma \ref{ym} for $s_1>s_2$ and $s_1=q_1 q_2 \cdots q_{n_1}$, then there exists an $s_2^n(\frac{s_1}{w}-1)$-QMDS code $((2d-1+n+n_1,s_2^{n},d+1))_{s_1^{2d-1} s_2^{n} q_1^1 q_2^1 \cdots q_{n_1}^1}$ for $w=max\{s_2,q_1,q_2,\ldots,q_{n_1}\}$.
\end{theorem}

\begin{proof}
By Lemma \ref{ym}, there exists an $OA(r,2d+n,s_1^{2d} s_2^n,t)$ with $MD=h$ and an orthogonal partition $\{A_1,A_2,\cdots, A_{s_2^n}\}$ of strength $t'$ and $d+1=min\{t'+1,h\}$. By Lemma \ref{kuozhang}, we can obtain an $OA(r,2d-1+n+n_1,s_1^{2d-1} s_2^n q_1^1 q_2^1 \cdots q_{n_1}^1,t)$ with $MD=h$ and an orthogonal partition $\{B_1,B_2,\ldots,B_{s_2^n}\}$ of strength $t'$. By Lemma \ref{ym} and Definition \ref{mQMDS}, there exists an $s_2^n(\frac{s_1}{w}-1)$-QMDS code $((2d-1+n+n_1,s_2^{n},d+1))_{s_1^{2d-1} s_2^{n} q_1^1 q_2^1 \cdots q_{n_1}^1}$ for $w=max\{s_2,q_1,q_2,\ldots,q_{n_1}\}$.
\end{proof}

\begin{example}
Let $4=q_1\times q_2=2\times2$. With the QMDS code $((7,8,3))_{4^4 2^3}$ in \cite{chao24} Theorem \ref{huan} produces an 8-QMDS code $((8,8,3))_{4^3 2^5}$ with eight basis states
$|\phi_1\rangle=
|0	0	0	0	0	0	0	0\rangle+
|0	1	1	0	1	0	1	1\rangle+
|0	2	2	1	0	1	0	1\rangle+
|0	3	3	1	1	1	1	0\rangle+
|1	0	1	1	0	1	1	0\rangle+
|1	1	0	1	1	1	0	1\rangle+
|1	2	3	0	0	0	1	1\rangle+
|1	3	2	0	1	0	0	0\rangle+
|2	0	2	1	1	0	1	1\rangle+
|2	1	3	1	0	0	0	0\rangle+
|2	2	0	0	1	1	1	0\rangle+
|2	3	1	0	0	1	0	1\rangle+
|3	0	3	0	1	1	0	1\rangle+
|3	1	2	0	0	1	1	0\rangle+
|3	2	1	1	1	0	0	0\rangle+
|3	3	0	1	0	0	1	1\rangle$,
$|\phi_2\rangle=
|0	0	0	0	0	0	0	0\rangle+
|0	1	1	0	1	0	1	1\rangle+
|0	2	2	1	0	1	0	1\rangle+
|0	3	3	1	1	1	1	0\rangle+
|1	0	1	1	0	1	1	0\rangle+
|1	1	0	1	1	1	0	1\rangle+
|1	2	3	0	0	0	1	1\rangle+
|1	3	2	0	1	0	0	0\rangle+
|2	0	2	1	1	0	1	1\rangle+
|2	1	3	1	0	0	0	0\rangle+
|2	2	0	0	1	1	1	0\rangle+
|2	3	1	0	0	1	0	1\rangle+
|3	0	3	0	1	1	0	1\rangle+
|3	1	2	0	0	1	1	0\rangle+
|3	2	1	1	1	0	0	0\rangle+
|3	3	0	1	0	0	1	1\rangle$,
$|\phi_3\rangle=
|0	0	1	1	1	0	1	0\rangle+
|0	1	0	1	0	0	0	1\rangle+
|0	2	3	0	1	1	1	1\rangle+
|0	3	2	0	0	1	0	0\rangle+
|1	0	0	0	1	1	0	0\rangle+
|1	1	1	0	0	1	1	1\rangle+
|1	2	2	1	1	0	0	1\rangle+
|1	3	3	1	0	0	1	0\rangle+
|2	0	3	0	0	0	0	1\rangle+
|2	1	2	0	1	0	1	0\rangle+
|2	2	1	1	0	1	0	0\rangle+
|2	3	0	1	1	1	1	1\rangle+
|3	0	2	1	0	1	1	1\rangle+
|3	1	3	1	1	1	0	0\rangle+
|3	2	0	0	0	0	1	0\rangle+
|3	3	1	0	1	0	0	1\rangle$,
$|\phi_4\rangle=
|0	0	1	1	1	1	0	1\rangle+
|0	1	0	1	0	1	1	0\rangle+
|0	2	3	0	1	0	0	0\rangle+
|0	3	2	0	0	0	1	1\rangle+
|1	0	0	0	1	0	1	1\rangle+
|1	1	1	0	0	0	0	0\rangle+
|1	2	2	1	1	1	1	0\rangle+
|1	3	3	1	0	1	0	1\rangle+
|2	0	3	0	0	1	1	0\rangle+
|2	1	2	0	1	1	0	1\rangle+
|2	2	1	1	0	0	1	1\rangle+
|2	3	0	1	1	0	0	0\rangle+
|3	0	2	1	0	0	0	0\rangle+
|3	1	3	1	1	0	1	1\rangle+
|3	2	0	0	0	1	0	1\rangle+
|3	3	1	0	1	1	1	0\rangle$,
$|\phi_5\rangle=
|0	0	2	0	1	0	0	1\rangle+
|0	1	3	0	0	0	1	0\rangle+
|0	2	0	1	1	1	0	0\rangle+
|0	3	1	1	0	1	1	1\rangle+
|1	0	3	1	1	1	1	1\rangle+
|1	1	2	1	0	1	0	0\rangle+
|1	2	1	0	1	0	1	0\rangle+
|1	3	0	0	0	0	0	1\rangle+
|2	0	0	1	0	0	1	0\rangle+
|2	1	1	1	1	0	0	1\rangle+
|2	2	2	0	0	1	1	1\rangle+
|2	3	3	0	1	1	0	0\rangle+
|3	0	1	0	0	1	0	0\rangle+
|3	1	0	0	1	1	1	1\rangle+
|3	2	3	1	0	0	0	1\rangle+
|3	3	2	1	1	0	1	0\rangle$,
$|\phi_6\rangle=
|0	0	2	0	1	1	1	0\rangle+
|0	1	3	0	0	1	0	1\rangle+
|0	2	0	1	1	0	1	1\rangle+
|0	3	1	1	0	0	0	0\rangle+
|1	0	3	1	1	0	0	0\rangle+
|1	1	2	1	0	0	1	1\rangle+
|1	2	1	0	1	1	0	1\rangle+
|1	3	0	0	0	1	1	0\rangle+
|2	0	0	1	0	1	0	1\rangle+
|2	1	1	1	1	1	1	0\rangle+
|2	2	2	0	0	0	0	0\rangle+
|2	3	3	0	1	0	1	1\rangle+
|3	0	1	0	0	0	1	1\rangle+
|3	1	0	0	1	0	0	0\rangle+
|3	2	3	1	0	1	1	0\rangle+
|3	3	2	1	1	1	0	1\rangle$,
$|\phi_7\rangle=
|0	0	3	1	0	0	1	1\rangle+
|0	1	2	1	1	0	0	0\rangle+
|0	2	1	0	0	1	1	0\rangle+
|0	3	0	0	1	1	0	1\rangle+
|1	0	2	0	0	1	0	1\rangle+
|1	1	3	0	1	1	1	0\rangle+
|1	2	0	1	0	0	0	0\rangle+
|1	3	1	1	1	0	1	1\rangle+
|2	0	1	0	1	0	0	0\rangle+
|2	1	0	0	0	0	1	1\rangle+
|2	2	3	1	1	1	0	1\rangle+
|2	3	2	1	0	1	1	0\rangle+
|3	0	0	1	1	1	1	0\rangle+
|3	1	1	1	0	1	0	1\rangle+
|3	2	2	0	1	0	1	1\rangle+
|3	3	3	0	0	0	0	0\rangle$ and
$|\phi_8\rangle=
|0	0	3	1	0	1	0	0\rangle+
|0	1	2	1	1	1	1	1\rangle+
|0	2	1	0	0	0	0	1\rangle+
|0	3	0	0	1	0	1	0\rangle+
|1	0	2	0	0	0	1	0\rangle+
|1	1	3	0	1	0	0	1\rangle+
|1	2	0	1	0	1	1	1\rangle+
|1	3	1	1	1	1	0	0\rangle+
|2	0	1	0	1	1	1	1\rangle+
|2	1	0	0	0	1	0	0\rangle+
|2	2	3	1	1	0	1	0\rangle+
|2	3	2	1	0	0	0	1\rangle+
|3	0	0	1	1	0	0	1\rangle+
|3	1	1	1	0	0	1	0\rangle+
|3	2	2	0	1	1	0	0\rangle+
|3	3	3	0	0	1	1	1\rangle$.
\end{example}

\section{Conclusion}
Quantum error-correcting codes theory deals with the problem of encoding quantum states into qudits such that a small number of errors can be detected, measured, and efficiently corrected. Wang et al. \cite{Wang13} build longer codes over mixed alphabets of distance 2. However, the existence and construction of QECCs $((n,K,d + 1))s_1s_2\cdots s_n$ for $d\geq2$ are still an open question \cite{Wang13}. In the paper, we introduce the definition of an $m$-QMDS code. We construct some new infinite classes of $m$-QMDS codes over mixed alphabets by using the construction methods of OAs such as deleting columns, orthogonal partitions, difference schemes and expansive replacement method. The codes obtained are listed in Tables I-VII, and are not limited to the classes listed. The results are not just existence results, but constructive results which provide a solid answer to the open problem.
These codes and methods also lay a foundation for constructing QMDS codes. In the future, we will study the construction of more QMDS and $m$-QMDS codes $((n,K,d + 1))s_1s_2\cdots s_n$ from asymmetrical OAs.\\
\textbf{Acknowledgments} This research was supported by NNSF of China Grant 11971004.\\
\textbf{Data Availability} Data sharing not applicable to this article as no datasets were generated or analysed during the current study.\\
\textbf{Conflict of interests} The authors have no competing interests to declare that are relevant to the content of this article.

\end{sloppypar}

\end{document}